\documentclass[journal,twoside,web]{ieeecolor}
\usepackage{generic}
\usepackage[noadjust]{cite}

\usepackage{amsmath,amssymb,amsfonts,amsthm}
\usepackage{algorithm}
\usepackage[algo2e]{algorithm2e}
\usepackage{algorithmic}
\usepackage{wasysym,xcolor}
\usepackage{hyperref}
\usepackage{mathrsfs}
\usepackage{graphicx,float}
\usepackage{capt-of}
\usepackage{textcomp,soul,comment}
\usepackage{pgfplots}
\usepackage{mathtools}
\usepackage{subcaption}

\usepackage{commands_and_packages}

\pgfplotsset{compat=1.10}
\def\BibTeX{{\rm B\kern-.05em{\sc i\kern-.025em b}\kern-.08em
    T\kern-.1667em\lower.7ex\hbox{E}\kern-.125emX}}
\newtheorem{definition}{Definition}
\newtheorem{theorem}{Theorem}
\newtheorem{proposition}{Proposition}
\newtheorem{assumption}{Assumption}
\newtheorem{lemma}{Lemma}
\newtheorem{corollary}{Corollary}

\allowdisplaybreaks

\begin{document}
\title{Time-Varying Coverage Control: \\A Distributed Tracker-Planner MPC Framework}
\author{ Patrick Benito Eberhard, Johannes Köhler, Oliver Hüsser, Melanie N. Zeilinger, Andrea Carron
\thanks{%
%For a visual demonstration of our contribution, we refer the reader to the video available under the following link: https://gitlab.ethz.ch/ics/time-varying-coverage-control.\\
Johannes K\"ohler was supported as a part of NCCR Automation, a National Centre of Competence in Research, funded by the Swiss National Science Foundation (grant number 51NF40\_225155).\\
All authors are members of the Institute for Dynamic Systems and Control (IDSC), ETH Z{\"u}rich. Corresponding author:
{\tt\footnotesize carrona@ethz.ch}}}

\maketitle
    
\begin{abstract}
Time-varying coverage control addresses the challenge of coordinating multiple agents covering an environment where regions of interest change over time. This problem has broad applications, including the deployment of autonomous taxis and coordination in search and rescue operations. The achievement of effective coverage is complicated by the presence of time-varying density functions, nonlinear agent dynamics, and stringent system and safety constraints. In this paper, we present a distributed multi-agent control framework for time-varying coverage under nonlinear constrained dynamics. Our approach integrates a reference trajectory planner and a tracking model predictive control (MPC) scheme, which operate at different frequencies within a multi-rate framework. For periodic density functions, we demonstrate closed-loop convergence to an optimal configuration of trajectories and provide formal guarantees regarding constraint satisfaction, collision avoidance, and recursive feasibility. Additionally, we propose an efficient algorithm capable of handling nonperiodic density functions, making the approach suitable for practical applications. Finally, we validate our method through hardware experiments using a fleet of four miniature race cars.
\end{abstract}
\noindent {\small {\bf Video}: \href{https://youtu.be/9kNvgjx3XbY}{https://youtu.be/9kNvgjx3XbY}}

\noindent {\small {\bf Code}: \href{https://gitlab.ethz.ch/ics/time-varying-coverage-control}
{https://gitlab.ethz.ch/ics/time-varying-coverage-control}}

\vspace{0.2cm}


\section{Introduction}
\label{sec:introduction} 
\textit{Coverage control } is an essential field in networked robotic systems and sensor networks, with applications ranging from environmental monitoring \cite{envmonitor}, search and rescue operations \cite{MACWAN}, robotic cleaning \cite{cleaning}, natural disaster mitigation \cite{oilspill}, and coordination of self-driving taxis \cite{carron2021}. The core objective of coverage control focuses on optimally allocating a network of agents to ensure efficient coverage of an environment, guided by a density function that reflects the spatial distribution of demands or priorities.

The density function is generally time-varying, reflecting the dynamic nature of the tasks and environments in which multi-agent systems operate, such as self-driving taxis, which must respond to a rapidly fluctuating demand. To deploy coverage control algorithms in the real world, the overall architecture should ensure collision avoidance of physical robots, consider time-varying environments, accommodate nonlinear dynamics, and facilitate scalable design and deployment. In this work, we develop a framework that effectively addresses these challenges.

\begin{figure}[t]
\centerline{\includegraphics[width=1.0\columnwidth]{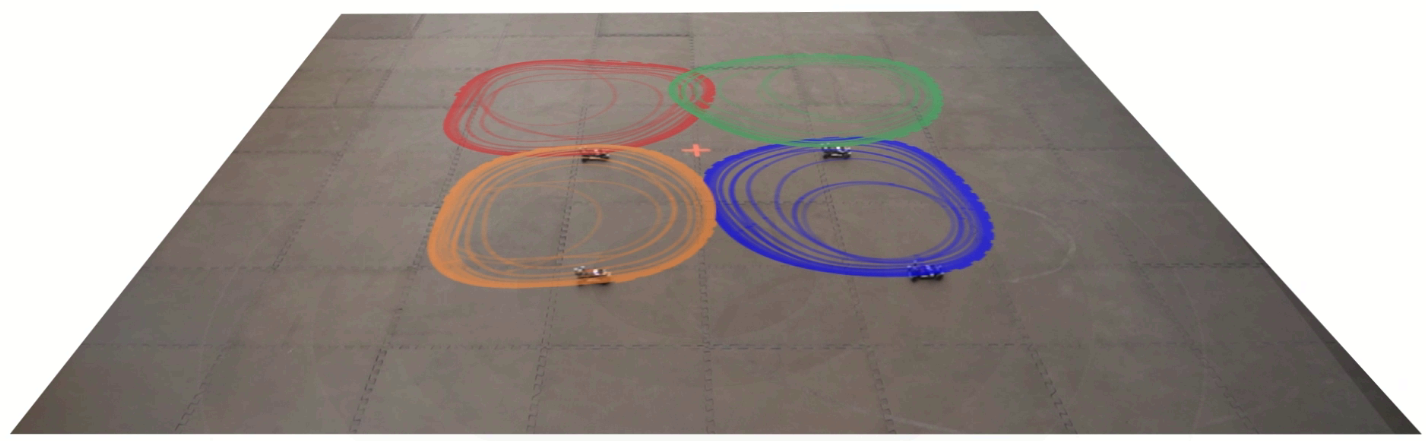}}
\caption{Closed-loop trajectories of four miniature cars covering a multivariate normal density function as it periodically translates in circles around the environment.}
\label{fig:experimental_setup}
\end{figure}

\subsubsection*{Related Work}
A variety of algorithms have been proposed to address coverage control problems. Most of them only consider a static density function and assume that each agent's motion follows single integrator dynamics. This assumption simplifies the problem, allowing agents to execute iterative solution approaches directly, as demonstrated in \cite{bullo,Corts2002CoverageCF}. In the latter work, the space is iteratively divided into optimal Voronoi partitions \cite{voronoi_tesselation_v2}, simplifying the multi-agent problem by solving a localized one within each agent's partition.

Algorithms for multi-agent coordination involving nonlinear constrained dynamics based on model predictive control (MPC) \cite{grune2017nonlinear, rawlings2017model} have been proposed in \cite{CARRONmpcc, KOHLER2018368, FARINA2015248}. The work in \cite{koehler_distributed_periodic_cooperation,kohler2025distributed} achieves cooperative control with a distributed MPC scheme that solves periodic artificial output trajectories. Recent developments, such as \cite{rickenbach2023}, incorporate artificial setpoints calculated within an MPC scheme for coverage control. However, this method does not accommodate time-varying densities, and the joint optimization of coverage planning and reference tracking can become computationally expensive. Such computational challenges can be addressed using hierarchical frameworks, where a reference planner and a tracker are executed at different frequencies~\cite{bender, koehler2020}. In a similar spirit, we develop a computationally efficient hierarchical framework for time-varying coverage control.

Significant progress in time-varying coverage control has been achieved for agents with integrator dynamics \cite{LEE2, LEE, santos}. These approaches have proven particularly effective in enabling human interaction with multi-robot swarms \cite{diaz}. Additionally, works such as \cite{Kennedy, Abdulghafoor2023MotionCO} incorporate collision avoidance mechanisms. However, the reliance on simplified integrator dynamics limits their applicability to nonlinear constrained systems, which cannot effectively consider arbitrary density functions. 

To conclude, approaches to time-varying coverage control with nonlinear constrained dynamics, which also guarantee persistent collision avoidance, remain largely unexplored in the literature. This work aims to address this gap through a unified and computationally tractable formulation.

\subsubsection*{Contributions}
This paper introduces a novel two-layer control architecture for time-varying coverage control involving nonlinear constrained dynamics. The algorithm comprises an optimal reference trajectory planner and a tracking MPC scheme that follows the planned references. The planning algorithm consists of an iterative method that computes sequences of spatial partitions and reference trajectories that adhere to the agent's dynamics, state, and input constraints. Additionally, we employ a tracking MPC scheme that does not require any terminal ingredients, circumventing their difficult design \cite{koehlerECO,grune2017nonlinear}. The two-layer architecture enables the tracker to operate at high frequencies while efficiently managing computational resources and decoupling the complex task of trajectory planning from real-time control. The algorithm operates independently on each agent in a distributed manner, where agents share their intended paths to determine new partition sequences through consensus. An overview of the architecture is shown in Figure~\ref{fig:algorithm_decoupled_coverage_control}.

We show closed-loop convergence properties to an optimal configuration of periodic trajectories and collision avoidance for periodic, time-varying density functions. Moreover, we devise an algorithm that can be applied to nonperiodic, time-varying density functions, offering collision avoidance guarantees.

To conclude, our approach is executed in a distributed manner, showcasing its robustness and scalability. We demonstrate the proposed approach on hardware using miniature racing cars in combination with CRS, an open-source software framework for control and robotics \cite{carron2023chronos}.

\subsubsection*{Outline}
Section~\ref{sec:problemFormulation} introduces the dynamics and constraints considered for each agent and the problem of optimal periodic coverage, and we propose an iterative approach in Section~\ref{section:optimal_coverage_planner}. Section~\ref{section:coverage_control_mpc} presents a two-layer algorithm based on an optimal periodic trajectory planner and a tracking MPC for addressing periodic coverage in practice. We further extend our framework to nonperiodic densities in Section~\ref{section:nonperiodic_coverage_control_mpc}. Section~\ref{section:discussion} provides a discussion on the proposed framework. We experimentally validate both algorithms on hardware in~\ref{section:results}, and an overall conclusion is provided in Section~\ref{section:conclusion}. Finally, some extensive proofs are deferred to the appendix.

\subsubsection*{Notation}
The Euclidean norm of a vector~$x\in\mathbb{R}^n$ is denoted by~$\|x\|$. The set of all non-negative real numbers is given by~$\mathbb{R}_+$, and the set of natural numbers by~$\mathbb{N}$. We define the ball~$\mathbb{B}_\epsilon^b = \{x \in \mathbb{R}^b \mid \|x\| \leq \epsilon\}$ and indicate with~$\ominus$ the Pontryagin set difference. By~$x_{k|t}$ we denote the prediction of a variable $x$ at time $t$ at $k$ steps in the future. Similarly,~$x_{a:b|t}$ denotes the prediction for~$k=a,\dots,b$, and~$x_{\cdot|t}$ denotes the complete prediction with length provided by the context. Further, $r_{\cdot+N|t}$ denotes the sequence~$r_{\cdot|t}$ shifted by~$N$ units. Finally, to denote an optimal solution to a minimization problem at the planner level for a variable $x$, we employ $\opt{x}$. Similarly, we use $\optt{x},\optt{u}$ to represent the optimal state and input computed by the tracker, respectively.

%Whenever we use the notation $r_t$ without indices, we mean the collection of the corresponding variable of all agents in the form of a vector, e.g., $p_t \colon [r_{1,t}, \dots, r_{M,t}]^\top$. Further, we use the notation $\rti$ to denote trajectories (and $\rt$ their collections) that start at time $t$ and have an arbitrary length given by the context. We employ $r_{\cdot+n|t}$ with $n \in \mathbb{N}$ to denote the trajectory $\rt$ shifted by $n$ units. We use $r_{a:b | t}$ for the subset of trajectories in an interval $[a,b]$ and starting at time $t$. To represent an optimal solution of a variable $r$, we employ $\opt{r}$.

\begin{comment}
\begin{figure}[!t]
\centerline{\includegraphics[width=1.0\columnwidth]{images/scheme.png}}
\caption{Graphical representation of the proposed architecture involving a tracking MPC and a trajectory planner. Agents exchange their reference trajectories, and recompute and validate new partition sequences in a distributed manner.}
\label{fig:algorithm_decoupled_coverage_control}
\end{figure}

\end{comment}

\section{Problem Formulation}
\label{sec:problemFormulation} 
In this section, we introduce the problem of coverage control and define the dynamics and constraints of the agents.

\subsection{Dynamics and Constraints}
Our task is to coordinate $M$ agents that move in a convex area $\mathbb{A} \subseteq \mathbb{R}^d$. Each agent $i \in \mathcal{M} := \{1,\ldots,M\}$ moves according to its own known dynamics.
\begin{subequations} \label{eq:agent_dynamics}
\begin{align}
    & x_{i, t+1} = f_i(x_{i,t}, u_{i,t}), \\
    & p_{i, t} = C_i x_{i,t},
\end{align}
\end{subequations}
where $x_{i, t} \in \mathbb{R}^{n_i}$ and $u_{i, t} \in \mathbb{R}^{m_i}$ are the state and input at time~$t$ for agent $i$ respectively. Furthermore, we impose constraints on the state and input, i.e., $x_{i, t} \in \X_i $ and $u_{i, t} \in \U_i$ for $t \in \mathbb{N}$. The matrix $C_i \in \mathbb{R}^{d \times n_i}$ allows us to extract the position~$p_{i,t} \in \mathbb{A}$ from the state vector. The latter is generally represented in Cartesian coordinates, i.e.,~$d=2$ or~$d=3$. It should be noted that our framework accommodates heterogeneous dynamics and constraints for each agent. In addition, we consider the following mild conditions.

\begin{assumption}
\label{assumption1}
The agent’s dynamics $f_i$ are known and Lipschitz continuous with Lipschitz constant $\mathcal{L}_{f_i}$.
Furthermore, the agent states can be perfectly measured, and the sets $\mathbb{X}_i$,$\mathbb{U}_i$, $\forall i\in\M$, and $\mathbb{A}$ are compact.
\end{assumption}

We additionally impose a minimum distance between all agents to avoid collisions. For all time instances $t \in \mathbb{N}$, it must be ensured that,
\begin{equation}
\label{collision_avoidance}
    \|p_{i,t} - p_{j,t}\| \geq 2R_{\max}, \: \forall i, j \in \M, i \neq j,
\end{equation}
where $R_{\max}$ represents the radius of the ball that covers the size of every agent.

\subsection{Coverage Control}
The task of coverage control is to optimally coordinate agents in the environment $\mathbb{A}$ with respect to a time-varying density function $\phi : \mathbb{A} \times \mathbb{N} \rightarrow \mathbb{R}_{+}$ that determines the importance of each point in the environment at a specific time. Additionally, the environment $\mathbb{A}$ is divided into a set of partitions $\W_{t} = \{ \W_{1,t}, \dots, \W_{M,t} \}$ that attribute an area of responsibility to each agent at a specific time $t$.   

We define the following coverage cost, which is also known as the \textit{locational cost}
\begin{equation} \label{eq:coverage_cost}
    H(p_t, \W_t, t) = \sum_{i = 1}^M \int_{\W_{i,t}} \mynorm{q - p_{i,t}}^2 \phi(q, t) dq,
\end{equation}
and is computed based on a fixed partition $\W_t$ and a set of positions $p_t= \{p_{1,t}, \ldots, p_{M,t} \}$  at time $t$. Every agent has a \textit{sensing capability} that is represented with the squared Euclidean norm. Intuitively, it suggests that the efficacy of the agent's actions diminishes with increasing distance from the point it is sensing.
Optimal coverage control involves finding a sequence of trajectories and partitions that minimize the cumulative locational cost~\eqref{eq:coverage_cost}. However, for general time-varying costs, computing optimal trajectories is only possible under additional conditions~\cite{economic_mpc_time_varying}. 
Hence, we also specifically study the problem where the density $\phi$ is periodic with some known period length $T$. 
In this case, the optimal solution to the periodic coverage problem can be formulated as 

\begin{subequations} \label{eq:coverage_problem}
\begin{align}
    \min_{x, u, \W} & \quad \sum_{t = 0}^{T-1} H(p_t, \W_t, t)\\
\label{eq:coverage_problem_periodicity_constraint}
\textrm{s.t.} & \quad  x_{i, 0} = x_{i, T} \\
& \quad  x_{i, t + 1} = f_i(x_{i,t}, u_{i, t}), \\
& \quad  x_{i,t} \in \X_i^{\intr}, u_{i,t}\in \U_i^{\intr}, \\
\label{eq:coverage_problem_position_constraint}
& \quad  p_{i,t} = C_i x_{i,t} \in \Wbar_{i, t}^{\intr}, \\
& \quad  t=0, \ldots, T-1,
\quad  i \in \M,\nonumber
\end{align}
\end{subequations}
where $\Wbar_i := \W_i \ominus \mathbb{B}_{R_{\max}}^2$ in~\eqref{eq:coverage_problem_position_constraint} ensures collision-free operation by tightening the partition of each agent with its radius $R_{\max}$. Due to the periodic nature of the problem, we impose a periodic constraint on the agent's states in \eqref{eq:coverage_problem_periodicity_constraint}. Considering some arbitrary small but fixed constant~$\epsilon > 0$, for each set~$\W \subseteq \mathbb{R}^b$ we denote~$\W^{\textrm{int}} = \W \ominus \mathbb{B}_\epsilon^b$. With this notation, $\X_i^{\intr}, \U_i^{\intr}, \Wbar_i^{\intr}$ represent sets that lie in the interior of their original sets, which will be used further to ensure that the physical agents can reach the planned trajectories. 

Overall, problem \eqref{eq:coverage_problem} represents a joint optimization of the positions $p_t$ and partitions $\W_t$ of all agents that minimize the coverage cost over an entire period. However, the joint nature of the optimization renders the problem computationally intractable.

\section{Optimal Periodic Coverage}
\label{section:optimal_coverage_planner}

Computing optimal coverage for \eqref{eq:coverage_problem} requires solving both agent trajectories and spatial partitions. To address this complexity, we propose an iterative approach that alternates between the optimization of agent trajectories and partitions, sharing similarities with Lloyd's algorithm \cite{lloyd}.

Consider the set of positions $p_t$ at a specific time $t$. The optimal partition of the area $\mathbb{A}$ based on the positions $p_t$ is given as the \textit{Voronoi tesselation}~\cite{voronoi_tesselation_v2}
\begin{equation} \label{eq:voronoi_partition}
    \W_{i, t}= \left\{q \in \mathbb{A} \mid\left\|q-p_{i,t}\right\| \leq\left\|q-p_{j,t}\right\|, \forall j \neq i\right\},
\end{equation}
where each point $q$ is assigned to the closest agent. 

For a time-varying density function, we require a sequence of partitions for the period length $T$, consisting of individual Voronoi partitions at each time step, i.e.
\begin{equation}
\label{eq:def_partition_seuqnce}
    \W_{\cdot | t} = \{\W_{0|t}, \W_{1|t}, \ldots, \W_{T-1|t} \},
\end{equation}
where $\cdot | t$ denotes a sequence starting from time $t$ covering the period $T$. Furthermore, we use $\W_{p_{\cdot | t}}$ to denote a partition sequence constructed with the position sequence $p_{\cdot | t}$, which is defined analogously to \eqref{eq:def_partition_seuqnce}.

We now address the problem of optimal periodic coverage planning by providing the following definition.
\begin{definition}
\label{def:optimal_periodic_configuration}
An optimal periodic configuration
is the collection of periodic trajectories~$r_{\cdot|t} = (x_{\cdot|t},u_{\cdot|t})$ and corresponding position sequence~$p_{\cdot|t}$ that are generators of the partition sequence~$\W_{p_{\cdot | t}}$ and are minimizers of~\eqref{eq:coverage_problem} for the same partitions~$\W_{p_{\cdot | t}}$.
\end{definition}

Note that the coverage cost for a specific agent $i$ and partition sequence $\W_{\cdot | t}$, starting at time $t$ for the period $T$ satisfies
\begin{equation} \label{eq:modified_coverage_cost}
\begin{aligned}
H_i^T&\left(p_{i,\cdot \mid t},\W_{i,\cdot \mid t}, t\right) \\
:=& \sum_{k=0}^{T-1} \int_{\W_{i,k \mid t}} \left\|q-p_{i, k \mid t}\right\|^2 \phi(q, t+k) d q, \\
 =&\sum_{k=0}^{T-1} m_{i, k \mid t} \,\mynorm{ \, p_{i, k \mid t}-c_{i, k \mid t}}^2 + d_{i,t},
 \end{aligned} 
\end{equation}
where $d_{i,t} \in \mathbb{R}$ is independent of the decision variables. The mass $m_{i,t}$ and centroid $c_{i,t}$ are defined as 
\begin{align}
    \label{eq:mass}
    m_{i,t} &= \int_{\W_{i,t}} \phi(q,t) d q,\\
    \label{eq:centroid}
    c_{i,t} &= \frac{1}{m_{i,t}}\int_{\W_{i,t}} q \phi(q,t) d q.
\end{align}
Given a partition sequence $\W_{i,\cdot \mid t}$, we can compute optimal trajectories $\opt{r}_{i,\cdot | t} \coloneqq (\opt{x}_{i,\cdot | t}, \opt{u}_{i,\cdot | t})$ independently for each agent $i \in \M$ by minimizing (\ref{eq:modified_coverage_cost}) while considering the agent's constraints, i.e., 
\begin{subequations}\label{eq:planner}
\begin{align}
\opt{r}_{i,\cdot | t} \coloneqq &\underset{x_{i,\cdot | t}, u_{i,\cdot | t}}{\argmin } \, H_i^T\left(p_{i,\cdot | t}, \mathbb{W}_{i,\cdot \mid t}, t\right) \\
\textrm{s.t.} \quad
\label{eq:planner2}
& x_{i, 0 \mid t} = x_{i, T \mid t},\\
\label{eq:planner1}
& x_{i, k + 1 \mid t} = f_i(x_{i,k \mid t}, u_{i,k \mid t}), \\
\label{eq:planner3}
& x_{i,k \mid t} \in \X_i^{\textrm{int}}, u_{i,k \mid t} \in \U_i^{\textrm{int}}, \\
\label{eq:planner4}
& p_{i,k \mid t} := C x_{i,k \mid t} \in \Wbar_{i,k \mid t}^{\textrm{int}}, \\
\label{eq:planner5}
& k=0, \ldots, T-1.
\end{align}
\end{subequations}
The solution of \eqref{eq:planner} yields an optimal\footnote{%
\label{footnote_nonunique}
In case the shifted candidate solution lies in the set of minimizers, we assume that this minimizer is chosen, i.e., $\opt{r}_{i,\cdot |t}=\opt{r}_{i,\cdot + 1|t-1}$.} trajectory $\opt{r}_{i,\cdot | t}$ and corresponding position sequence $\opt{p}_{i,\cdot | t}$ for a specific agent $i$ concerning its partition sequence $\mathbb{W}_{i,\cdot \mid t}$. Since the solution is periodic, we can define the following shifting operation for periodic references and partition sequences
\begin{subequations}
\begin{align}
     \label{eq:shifted_reference_periodic}
     r_{\cdot + n|t} &\coloneqq \{ r_{n|t}, \dots,
        r_{T-1|t}, r_{0|t} \dots, r_{n-1|t}  \},\\
    \label{eq:shifed_partition_periodic}
     \mathbb{W}_{\cdot + n|t} &\coloneqq \{ \mathbb{W}_{\sub[n]}, \dots , 
     \mathbb{W}_{\sub[T-1]},  
     \mathbb{W}_{\sub[0]}, \dots,  \mathbb{W}_{\sub[n-1]} \}.
\end{align}
\end{subequations}
We propose an optimal periodic coverage planner in Algorithm~\ref{alg:planner_convergence}, which performs an update of the partition sequences based on the optimal positions $\opt{p}_{\cdot \mid t}$ using \eqref{eq:voronoi_partition}, and then recomputes the trajectories, repeating this process iteratively until convergence. 
%We note that $\opt{p}_{i,\cdot + 1|t}$, i.e. the one-step shifted solution from the previous iteration, is used if the coverage cost decrease for the minimizer is not strict. 
% In each iteration, we shift the partition sequences and trajectories of length $T$ by one time step, represented with $\opt{p}_{\cdot+1|t}$ and $\W_{i+1|t}$. 
Moreover, it follows from the periodicity of the problem that $\opt{p}_{i, \cdot+T|t} = \opt{p}_{i, \cdot|t}$ and $\W_{i,\cdot+T|t} = \W_{i, \cdot|t}$, which is exploited in the algorithm. 
\begin{algorithm}[H]
  \caption{Optimal Periodic Coverage}
  \label{alg:planner_convergence}
  \SetAlgoLined
  \KwIn{Any $\W_{\cdot|0}$ and $\opt{r}_{\cdot | 0}$, $\opt{p}_{\cdot | 0}$ satisfying constraints in \eqref{eq:planner}}\{\textit{Initial trajectory and position sequence}\}\\
  \For{$t = 0, 1, 2, \dots$}{
      Set $\mathbb{W}_{\cdot|t} = \W_{\opt{p}_{\cdot|t}}$ using \eqref{eq:voronoi_partition} with $\opt{p}_{\cdot|t}$ \\\{\textit{Update partition sequence}\} \\
      Compute $r_{i,\cdot|t+1}^\star$ by solving \eqref{eq:planner} with $\mathbb{W}_{i,\cdot+1|t}$ for $i \in \M$ \{\textit{Compute optimal trajectories}\}
  }
\end{algorithm}
Note that the initial trajectory and position sequence can be chosen as any steady-state condition for each agent. 

Furthermore, we state the following lemma on the generators of a partition sequence.
\begin{lemma}\cite[Lemma 1.]{rickenbach2023} \label{lemma:feasible_partition_update}
    For any $\Bar{p} \in \mathbb{A}^{M}$ that satisfies $\mynorm{\Bar{p_i} - \Bar{p_j}} \geq 2(R_{\max} + \epsilon)$ for all $j \neq i$, it holds that $\Bar{p} \in \W_{\Bar{p}}^{\intr}$. Moreover, $\forall p$ with $\mynorm{\Bar{p} - p} \leq \epsilon$, it holds $p \in \W_{\Bar{p}}$.
\end{lemma}
This implies that any trajectory satisfying the collision avoidance constraint \eqref{collision_avoidance} will be contained in the interior of its Voronoi partition sequence. Furthermore, the lemma ensures that an arbitrarily close point $p$ to the generator $\Bar{p}$ is contained in the partition $\W_{\Bar{p}}$.

To establish the convergence of Algorithm~\ref{alg:planner_convergence} to an optimal periodic configuration (Def. \ref{def:optimal_periodic_configuration}), we introduce the following assumption, commonly considered in coverage control~\cite{Corts2002CoverageCF}.

\begin{assumption} \label{assumption:fintie_centroidal_periodic_configuration}
The set of optimal periodic configurations is finite.
\end{assumption} 

The convergence properties of Algorithm~\ref{alg:planner_convergence} are formalized in the following proposition:
\begin{proposition} \label{proposition:planner_convergence}
%Consider any initial partition sequence $\mathbb{W}_{\cdot|t}$ at $t=0$ based on a periodic trajectory $r_{\cdot|t}$ satisfying the constraints in ~\eqref{eq:planner}. 
%Consider an iteration where for every $t\in\mathbb{N}$, $\mathbb{W}_{\cdot|t} = \W_{\opt{p}_{\cdot,t}}$ is the optimal Voronoi tessellation \eqref{eq:voronoi_partition} w.r.t. the periodic position trajectories $p_{\cdot|t}^\star$. Subsequently, $p_{i,\cdot|t+1}^\star$ is a minimizer of Problem \eqref{eq:planner} for each agent $i$ with partition sequence $\mathbb{W}_{i,\cdot+1|t}$. 

 Suppose that Assumption~\ref{assumption:fintie_centroidal_periodic_configuration} is satisfied. Then, Algorithm~\ref{alg:planner_convergence} is feasible for all $t\in\mathbb{N}$ and the sequences $\mathbb{W}_{\cdot|t}$, $r^\star_{\cdot|t}$ converge to an optimal periodic configuration (Def. \ref{def:optimal_periodic_configuration}) as $t\rightarrow\infty$.
\end{proposition}

\begin{proof} The proof builds upon the proof in \cite[Prop. 3.3]{Corts2002CoverageCF}.
We first verify that $p_{\cdot|t}^\star$ remains feasible in $\W_{\opt{p}_{i,\cdot|t}}$, i.e., the position sequence satisfies the constraints also after a partition update. This condition is guaranteed for any trajectory that respects the constraints in \eqref{eq:planner}, as established by Lemma~\ref{lemma:feasible_partition_update} on the properties of the Voronoi tessellation \eqref{eq:voronoi_partition}. Note that Algorithm~\ref{alg:planner_convergence} is initialized with a feasible solution to Problem~\eqref{eq:planner}, and hence the optimization problems in Algorithm~\ref{alg:planner_convergence} are feasible for all $t\in\mathbb{N}$. 
Furthermore, Assumption~\ref{assumption1} ensures that all variables lie in compact sets. 
%Therefore, with Assumption~\ref{assumption1}, the vehicle lies in a compact, positively invariant set.

%First, it must be verified that $p_{\cdot|t}^\star$ is feasible in $\W_{\opt{p}_{i,\cdot|t}}$, i.e. it is still valid after a partition update. This condition always holds for each trajectory respecting constraints in \eqref{eq:planner} as per construction of the Voronoi tesselation \eqref{eq:voronoi_partition} (see Lemma~\ref{lemma:feasible_partition_update}). Thus, the iteration is recursively feasible provided that the initial trajectories respect the constraints in \eqref{eq:planner}.

%Moreover, we show that the coverage cost decreases at each iteration and represents a descent function, i.e.
Next, we show that
\begin{equation}
    H^T\left(p_{\cdot\mid t},\mathbb{W}_{\cdot \mid t}, t\right) := \sum_{i=1}^{M} H_i^T\left(p_{i, \mid t},\mathbb{W}_{i,\cdot \mid t}, t\right),
\end{equation}
is a descent function for Algorithm~\ref{alg:planner_convergence}, i.e., the coverage cost decreases at each iteration. Specifically, it holds that
\begin{subequations}
\begin{align}
    \begin{split}
    & H^T(\opt{p}_{\cdot|t}, \W_{\opt{p}_{\cdot+1|t-1}}, t)\\
    &\geq H^T(\opt{p}_{\cdot|t}, \W_{\opt{p}_{\cdot|t}}, t) \label{eq:coverage_descent_update_partition}\\ 
    \end{split}\\
    \label{eq:coverage_descent_equal}
    &= H^T(\opt{p}_{\cdot+1|t}, \W_{\opt{p}_{\cdot + 1|t}}, t + 1)\\
    &\geq H^T(\opt{p}_{\cdot|t+1}, \W_{\opt{p}_{\cdot + 1|t}}, t + 1) \label{eq:coverage_descent_minimizer}\\
    &\geq 0, \forall \, t \in \mathbb{N} \label{coverage_descent_upper_bound}.
\end{align}
\end{subequations}
Inequality~\eqref{eq:coverage_descent_update_partition} follows from the properties of the Voronoi tessellation \eqref{eq:voronoi_partition}: every partition update reduces the overall coverage cost since each Voronoi partition $\W_{\opt{p}_{j|t}}$ is optimal for the corresponding set of positions ${p}_{j|t}$, for $j \in [0,\ldots, T-1]$  \cite{Corts2002CoverageCF}. Consequently, the coverage cost strictly decreases for any partition in the sequence $\W_{\opt{p}_{\cdot|t}}$ that does not correspond to a Voronoi tessellation.
%The relation in \eqref{eq:coverage_descent_update_partition} can be shown as follows: every partition update with \eqref{eq:voronoi_partition} can only lead to a decrease in the overall coverage cost since each Voronoi partition $\W_{\opt{p}_{j|t}}$ is the optimal partition for the set of positions ${p}_{j|t}$, for $j = 0,\ldots,T-1$  \cite{Corts2002CoverageCF}. Consequently, the coverage cost decreases for each element in the partition sequence $\W_{\opt{p}_{\cdot|t}}$, and this decrease is strict if any element in the partition sequence does not correspond to a Voronoi tesselation.
%Furthermore, Inequality \eqref{eq:coverage_descent_minimizer} holds since  $p_{\cdot|t}^\star$ is feasible in the new partition sequence $\W_{\opt{p}_{i,\cdot|t}}$, and due to the periodicity constraint, the shifted trajectory $r_{i,\cdot|t+1} = \opt{r}_{i,\cdot + 1|t}$ is a candidate solution that provides an upper bound to the coverage cost for the minimizer of \eqref{eq:planner} for each agent $i \in \M$. Furthermore, we use $\opt{p}_{i,\cdot + 1|t}$ if the coverage cost decrease of this minimizer is not strict.
Equality~\eqref{eq:coverage_descent_equal} uses shift invariance of the coverage cost. 
Inequality~\eqref{eq:coverage_descent_minimizer} holds because $p_{\cdot|t}^\star$ remains feasible in the new partition sequence $\W_{\opt{p}_{i,\cdot|t}}$. Under the periodicity constraint, the shifted trajectory $r_{i,\cdot|t+1} = \opt{r}_{i,\cdot + 1|t}$ serves as a candidate solution, providing an upper bound to the coverage cost of the minimizer of \eqref{eq:planner} for each agent $i \in \M$. 
Furthermore, \eqref{coverage_descent_upper_bound} is satisfied since $\phi(q) \in \mathbb{R}_{+} \, \forall q \in \mathbb{A}$. Finally, note that we select the minimizer $p^\star_{\cdot|t+1}=p^\star_{\cdot+1|t}$ in case the cost remains constant (see Footnote~\ref{footnote_nonunique}). 

Under these conditions, the sequences $\mathbb{W}_{\cdot|t}$, $r^\star_{\cdot|t}$ converge to an optimal periodic configuration (Def.~\ref{def:optimal_periodic_configuration}) as $t\rightarrow\infty$ \cite[Lemma 1.3, Prop. 1.4]{Corts2002CoverageCF}.
\end{proof}

Note that for a period length $T=1$, and assuming that all positions accept a feasible steady-state, Algorithm~\ref{alg:planner_convergence} reduces to Lloyd's algorithm~\cite{lloyd}.

\section{Periodic Coverage MPC}
\label{section:coverage_control_mpc}
The previous section provides an effective method for computing periodic trajectories that converge to an optimal periodic configuration. However, this method solely focuses on trajectory planning and does not address real-time feedback required to track the reference trajectory. 
Note that the periodic planning problem is computationally complex and can typically not be solved at the fast update rates required for real-time feedback. 
%Note that the computational complexity of the periodic planning algorithm prevents the algorithm from operating at the fast update rates needed for real-time control.

To tackle this problem, we propose a novel algorithm comprising a reference trajectory planner and a tracking MPC, operating at different frequencies in a multi-rate scheme. The proposed method decouples the planning and tracking tasks, enhancing both computational efficiency and real-time applicability in multi-agent systems. Ultimately, our approach is implemented in a distributed manner.

\subsection{Control Architecture}
\label{subsection:algorithm}
This section presents the proposed control architecture for optimal periodic coverage, which consists of a reference trajectory planner and a tracking MPC.

The reference trajectory planner is executed every $K$ steps, where $K \in \mathbb{N}$ is a user-defined constant. The planner computes optimal and reachable reference trajectories by solving an optimization problem similar to \eqref{eq:planner} for a given partition sequence $\W_{i,\cdot|t}$. Importantly, we will introduce a \textit{coupling constraint} that links the reference planner with the tracking MPC, ensuring that each computed reference remains feasible within the tracking controller’s limitations.

At every time $t$, the tracking MPC aims to follow the planner’s reference trajectory and applies an optimal control input to the system, while ensuring that the vehicle's motion respects all constraints.

The reference planner and tracking MPC operate locally and in parallel with the agent’s central processing unit. Communication occurs every $K$ steps to broadcast the agent-specific reference trajectories, which are required to recompute the partition sequences. Thereafter, a consensus is established across agents to ensure compatibility with the newly computed partition sequence, promoting both collision avoidance and mutual feasibility across all agents’ trackers. Finally, the process is repeated until convergence.

The complete control architecture is provided in Algorithm \ref{alg:decoupled_coverage_control} and illustrated in Figure \ref{fig:algorithm_decoupled_coverage_control}. At initialization, each agent $i$ is assumed to start in a steady-state $x_{i,0} = f_i(x_{i,0}, u_{i,0}) \in \X^{\intr}_i$ with $u_{i,0} \in \U^{\intr}_i$ for all $i \in \M$. In addition, all agents are located at least~$2(R_{\max} + \epsilon)$ apart. This ensures that each agent begins from a feasible operating point and can respect the collision avoidance constraints.

\begin{algorithm}
  \caption{Periodic Coverage MPC}
  \label{alg:decoupled_coverage_control}
  \SetAlgoLined
  \KwIn{$\phi$, $x_{i,0} = f_i(x_{i,0}$, $u_{i,0}) \in \X^{\intr}, \;  u_{i,0} \in \U^{\intr}, \; \forall i \in \M$\\
  $\quad \quad \quad\|p_{i,0} - p_{j,0}\| \geq 2(R_{\max} + \epsilon), \quad \forall i \neq j \in \M$}
  
    Set $\mathbb{W}_{\sub} = \{\mathbb{W}_{p_0}, \dots, \mathbb{W}_{p_0}\}$
    
    Compute $\opt{r}_{i,\cdot|0}$ with \eqref{eq:periodic_reference_planner} \{\textit{Initial reference planning}\}
    
        \For{$t = 0, 1, \dots$}{
            \If{t mod K = $0$}{
            Construct candidate $\hat{x}_{i, \cdot | t}, \hat{u}_{i, \cdot | t}$ with \eqref{eq:fast_coverage_candidate_sequence}
            
            \If{Candidates satisfy \eqref{eq:partition_update_condition} $\forall i \in \M$}{
                    \{\textit{Consensus}\}
                    
                    Set $\W_{\sub} = \W_{\opt{p}_{\cdot|t}}$ \{\textit{Partition update}\}
                }
                Compute $\opt{r}_{i,\cdot|t + K}$ with \eqref{eq:periodic_reference_planner} \{\textit{Reference planning}\}
                
                Communicate $\opt{p}_{i,\cdot|t + K}$ to other agents
                
                Compute $\W_{\opt{p}_{\cdot|t + K}}$ with \eqref{eq:voronoi_partition}
            }
            Solve MPC \eqref{eq:tracking_feasible_trajectory} with $\opt{r}_{i,\cdot|t}, \W_{i,\cdot|t}$ to obtain $\optt{\uti}$ \\\{\textit{Tracking MPC}\}
            
            Apply $u_{i,t}=\optt{\uti[0]}$ 

            $\opt{r}_{i,\cdot|t+1} = \opt{r}_{i,\cdot + 1|t}$ \eqref{eq:shifted_reference_periodic}, $\W_{i,\subplus} = \W_{i,\cdot + 1|t}$ \eqref{eq:shifed_partition_periodic}
            
            \{\textit{Shift reference and partitions}\}
        }        
\end{algorithm}

\begin{figure}[!t]
\centerline{\includegraphics[width=1.0\columnwidth]{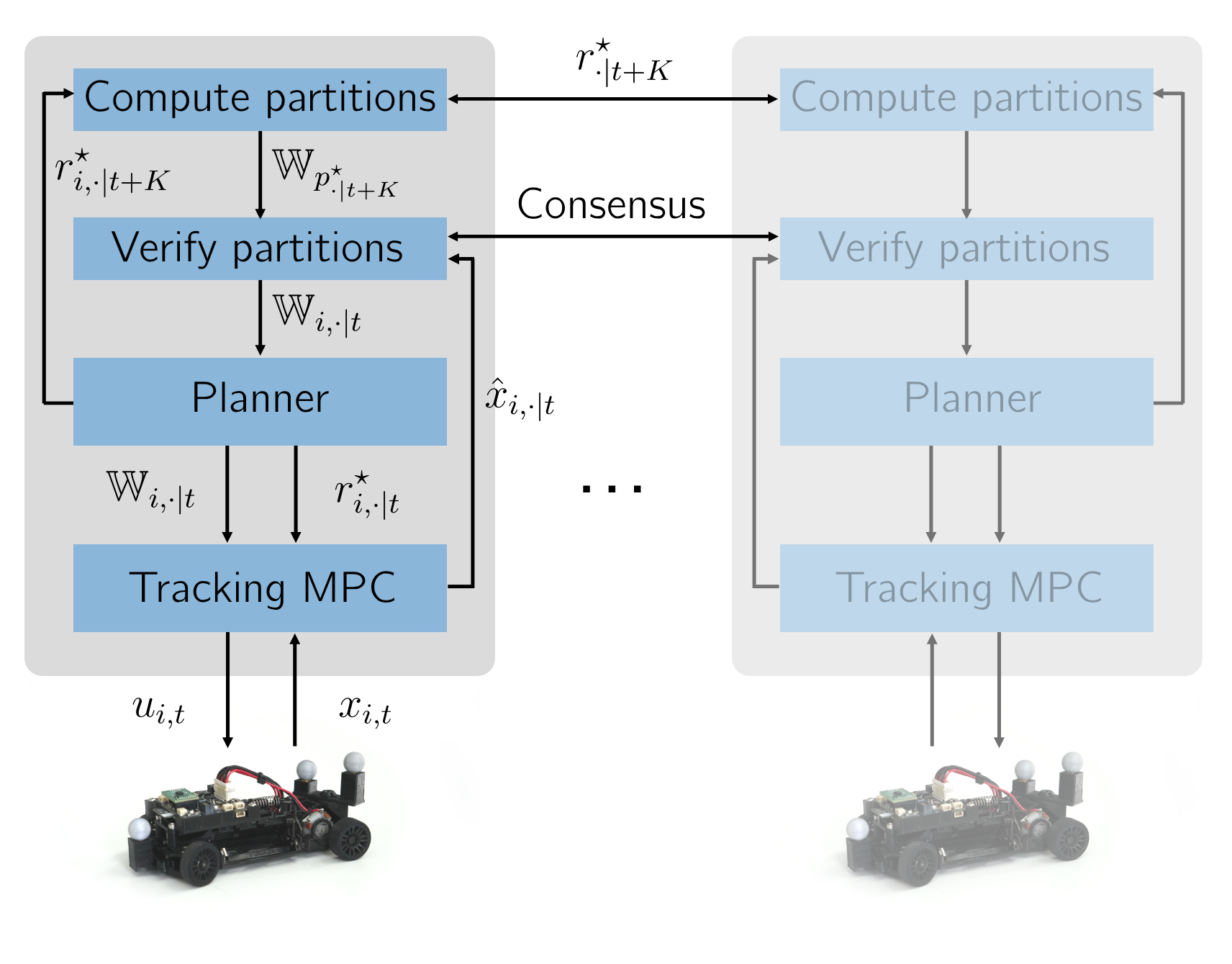}}
\caption{Graphical illustration of the proposed architecture integrating a tracking MPC with a reference trajectory planner. Every $K$ time steps, agents compute new reference trajectories $\opt{r_{i,\cdot|t+K}}$, exchange them to construct the candidate partitions $\W_{\opt p_{\cdot|t+K}}$, and reach consensus on admissible partition sequences $\W_{i,\cdot|t}$ by verifying the feasibility of the candidate trajectories $\hat{x}_{i,\cdot|t}$. The tracking MPC uses the partitions $\W_{i,\cdot|t}$, reference $\opt{r_{i,\cdot|t}}$, and current state $x_{i,t}$ of the vehicle to compute and apply the control input $u_{i,t}$ at every time step.}
\label{fig:algorithm_decoupled_coverage_control}
\end{figure}

\subsection{Nonlinear Tracking MPC}
In this section, we introduce a tracking MPC that enables the system to follow a given reference trajectory. We leverage a design without terminal ingredients, avoiding their challenging design requirements \cite{koehlerECO,grune2017nonlinear}. 

We first introduce a definition for reachable references that can be tracked effectively:

\begin{definition}
\label{def:reachable_reference}
    A reference partition sequence $\W_{\cdot|t}$ and a reference trajectory $r_{\cdot | t} = (x^r_{\cdot | t},u^r_{\cdot | t})$ are considered to be \textit{reachable} if they satisfy the constraints in \eqref{eq:planner1}, \eqref{eq:planner3} and~\eqref{eq:planner4} for $k \in \mathbb{N}$.
\end{definition}

Consider a horizon of length $N_i$ and the tracking cost function
\begin{equation}
    \JNi(\xti, \uti, \rti) = \sum_{k=0}^{N_i-1} \ell_i(\xti[k], \uti[k], \rti[k]),
\end{equation}
where the cost $\JNi$ is composed of the sum of the stage costs $\ell_i$ over the finite horizon $N_i$. In this article, we will focus on quadratic-stage costs for simplicity of exposition
\begin{equation}
    \ell_i(x_i, u_i, r_i) = \mynorm{x_i - x_i^r}^2_{Q_i} + \mynorm{u_i -  u_i^r}^2_{R_i},
\end{equation}
where $Q_i,R_i \succ 0$ and $r_i = (x_i^r,u_i^r) \in \X_i \times \U_i$ and 
\begin{equation} \label{eq:bounds_on_stage_cost}
        \alpha_{1,i}\mynorm{x_i - x_i^r}^2 \leq \mynorm{x_i -x_i^r}^2_{Q_i} \leq \alpha_{2,i}\mynorm{x_i - x_i^r}^2, 
\end{equation}
with $\alpha_{1,i} = \lambda_{\text{min}}(Q_i)$, $\alpha_{2,i} = \lambda_{\text{max}}(Q_i)$. 
%The proposed framework can be naturally extended to non-quadratic stage cost functions, similar to~\cite{rickenbach2023}.

\begin{comment}
Furthermore,  
\begin{equation}
\label{eq:l_star}
\opt{\ell_i}(x_i, r_i) \coloneqq \ell_i(x_i, u_i^r, r_i), \:\: \text{where} \: \opt{u}_i = u_i^r,
\end{equation}
i.e. the stage cost is trivially minimized by the reference input.
\end{comment}
\begin{comment}
\begin{assumption} 
\label{assumption:stage_cost}
    There exist constants $\alpha_{1,i}, \alpha_{2,i} > 0$ such that the following holds for all $x_i \in \X_i, u_i \in \U_i$ and any reference point $r_i = (x_i^r, u_i^r)  \in \X_i \times \U_i$:
    \begin{equation*} \label{eq:bounds_on_stage_cost}
        \alpha_{1,i}\mynorm{x_i - x_i^r}^2 \leq \opt{\ell_i}(x_i, r) \leq \alpha_{2,i}\mynorm{x_i - x_i^r}^2 
    \end{equation*}
    where 
    \begin{align} \label{eq:one_step_optimal_stage_cost}
        \opt{\ell_i}(x_i, r_i) \coloneqq \min_{u_i} \quad & \ell_i(x_i, u_i, r_i) \quad 
        \textrm{s.t.} \quad  x_i \in \X_i, u_i \in \U_i.
    \end{align}

    is minimized by the input $\opt{u}_i = u_i^r$ .

    Additionally, there exist $k_{0,i}^l, k_{1,i}^l > 0$ such that for any two reachable references $\rti, \hat{r}_{\subi}$, it holds
    \begin{equation*}
        \ell_i(x_i, u_i, \rti[k]) \leq k_{0,i}^l \ell_i(x_i, u_i, \hat{r}_{\subi[k]}) + k_{1,i}^l \mynorm{\rti[k] - \hat{r}_{\subi[k]}}^2.
    \end{equation*}
\end{assumption}
\end{comment}

The following tracking MPC formulation is used to track the reachable reference $\rti$ for agent $i$.
\begin{subequations} \label{eq:tracking_feasible_trajectory}
\begin{align} 
    \VNi(x_{i,t}, \rt, \W_{i,\cdot|t}) =& \, \min_{\uti} \JNi(\xti, \uti, \rti) \\
    \textrm{s.t.} \quad & \xti[0] = x_{i,t},  \\
    & \xti[k+1] = f(\xti[k], \uti[k]), \\
    \label{mpc_state_and_input_constraints}
    & \xti[k] \in \X_i, \uti[k] \in \U_i  \\ 
    \label{mpc_partition_constraints}
    & p_{i,k \mid t} := C x_{i,k \mid t} \in \Wbar_{i,k \mid t},\\ 
    & k = 0, \ldots, N_i - 1.
\end{align}
\end{subequations}
Solving \eqref{eq:tracking_feasible_trajectory} yields an input trajectory $\optt{\uti}$ and a corresponding state trajectory $\optt{\xti}$ of length $N_i$. The loop is closed by applying the first input of $\optt{\uti}$, i.e.
\begin{equation} \label{eq:closed_loop_standard_tracking}
    x_{i,t+1} = f(x_{i,t}, \optt{\uti[0]}).
\end{equation}

We further require a local stabilizability condition \cite[Assumption 1]{grune2017nonlinear} for the discussion on the stability properties of the tracking MPC.
\begin{assumption} \label{assumption:local_controllability}
    There exist $\gamma_i, c_i > 0$ such that for any \mbox{$N_i \in \mathbb{N}$}, any reachable reference $\rti$, and any $x_{i,t}$ with \mbox{$\mynorm{x_{i,t}-x_{i,0|t}^r}^2_{Q_i} \leq c_i$}, we have 
    \begin{equation*} \label{eq:local_controllability}
        \VNi(x_{i,t}, \rti) \leq \gamma_i \mynorm{x_{i,t} -x_{i,0|t}^r}^2_{Q_i} .
    \end{equation*}
\end{assumption}
Note that Assumption \ref{assumption:local_controllability} is valid if the linearized
dynamics are stabilizable \cite{grune2017nonlinear}.

Furthermore, the following theorem summarizes the properties of the tracking MPC

\begin{theorem} \label{theorem:stability_of_tracking_mpc}
    (\cite[Thm.~1]{koehlerECO})
    Let Assumption~\ref{assumption:local_controllability} hold, suppose that $r_{i,\cdot|t}$ is a reachable reference (Def.~\ref{def:reachable_reference}) with $r_{i,\cdot|t+1}:=r_{i,\cdot+1|t}$, and consider an initial condition satisfying $\VNi(x_{i,t}, \rti) \leq V_{\max,i}$ at $t=0$. Then, for any constants $V_{\max,i}>0$, $\alpha_{N,i} \in (0,1)$, $\Bar{\gamma}_i \coloneqq \max\left\{\gamma_i, \frac{V_{\max,i}}{c_i}\right\}$, there exists a horizon $N_{0,i}\in\mathbb{N}$, such that for all $N_i\geq N_{0,i}$ Problem~\eqref{eq:tracking_feasible_trajectory} is recursively feasible, the closed-loop system~\eqref{eq:closed_loop_standard_tracking} converges exponentially to the reachable reference $\rti$, and the value function satisfies
    \begin{subequations}
    \begin{align}
    \begin{split} \label{eq:smth}
        &\VNi(x_{i,t}, \rti) \leq \Bar{\gamma}_i \mynorm{x_{i,t} - x_{i,0|t}^r}^2_{Q_i},\\
    \end{split}\\
    \begin{split} \label{eq:lyapunov_decrease_standard_tracking}
     &\VNi(x_{i,t+1}, \rtiplus) - \VNi(x_{i,t}, \rti) \\
    & \quad \quad \leq - \alpha_{N,i} \ell_i(x_{i,t}, \opt{\uti[0]}, \rti[0]).
    \end{split}
    \end{align}
    \end{subequations}
%    Let Assumption~\ref{assumption:local_controllability} hold, and consider any constants $V_{\max,i}>0$, $\Bar{\gamma}_i \coloneqq \max \{\gamma_i, \frac{V_{\max,i}}{c_i}\}$ and $\alpha_{N,i} \in (0,1)$. Further, there exists an $N_{0,i} \in \mathbb{N}$ such that for all $N_i \geq N_{0,i}$ it holds that $\VNi(x_{i,t}, \rti) \leq V_{\max,i}$ at $t=0$.
%    Then, problem \eqref{eq:tracking_feasible_trajectory} is recursively feasible, and the closed-loop system~\eqref{eq:closed_loop_standard_tracking} converges exponentially to the reachable reference $\rti$.     
\end{theorem}

%From the previous theorem, it follows that
%    \begin{align*}
%        & \alpha_{1,i} \mynorm{x_{i,t} - x_{i,t}^r}^2 \leq \VNi(x_{i,t}, \rt) \leq \Bar{\gamma}_i \alpha_{2,i} \mynorm{x_{i,t} - x_{i,t}^r}^2,
%    \end{align*}
%    where the lower bound is derived with \eqref{eq:bounds_on_stage_cost} and the definition of $\JNi$
%    \begin{align}
%        & \VNi(x_{i,t}, \rti) \geq \ell_i(x_{i,t}, \opt{\uti[0]}, \rti[0]) \\
%        & \quad \geq \opt{\ell_i}(x_{i,t}, \rti[0]) \geq \alpha_{1,i} \mynorm{x_{i,t} - x_{i,t}^r}^2,
%    \end{align}
%    while the upper bound follows from \eqref{eq:bounds_on_stage_cost} and Assumption  \ref{assumption:local_controllability}
%    \begin{equation}
 %       \VNi(x_{i,t}, \rti) \leq \Bar{\gamma_i} \opt{\ell_i}(x_{i,t}, \rti[0]) \leq \Bar{\gamma} \alpha_{2,i} \mynorm{x_{i,t} - x_{i,t}^r}^2.
 %   \end{equation}
 %   Moreover, it holds that 
 %   \begin{align}
 %   \begin{split}
 %       & \quad  \quad  \VNi(x_{i,t+1}, \rtiplus) - \VNi(x_{i,t}, \rti)\\ 
 %       &\quad \quad \quad \leq - \alpha_{N,i} \ell_i(x_{i,t}, %\opt{\uti[0]}, \rti[0]),\label{eq:lyapunov_decrease_standard_tracking}
        %\end{split}
%    \end{align}
%    where 
The constants $\alpha_{N,i} $ and $N_{0,i}$ can be found in \cite{koehlerECO}.

\subsection{Periodic Coverage Planner}
\label{section:decoupled_coverage_planner}
It is important to note that running the reference planner independently of the tracking MPC results in references that do not satisfy the conditions outlined in Theorem~\ref{theorem:stability_of_tracking_mpc}, thereby losing all closed-loop guarantees of the tracking MPC. For this reason, we leverage the following proposition.
\begin{proposition}
\label{proposition:coupling_constraint}
Let Assumption~\ref{assumption1} hold and consider the closed-loop system according to Algorithm~\ref{alg:decoupled_coverage_control}. There exists a constant $\LipV>0$ such that, given a reachable reference $\opt{r}_{\subi}$ and partition sequence $\mathbb{W}_{i, \sub}$, any new reachable reference trajectory $\opt{r}_{\subiplusk}$ satisfying 
\begin{align} \label{eq:periodic_coupling_constraint_derivation}
\begin{split}
    & \mynorm{\opt{r}_{\subi[K:K+N]} - \opt{r}_{\subiplusk[0:N]}}\\
    &\quad\leq  \frac{1}{\LipV} \left( V_{\max, i} - \left( 1 - \frac{\alpha_{N,i}}{\gammabari} \right)^K \VNit  \right) := \C(\VNit),
\end{split}
\end{align}
with $\VNit = \VNi(x_{i, t}, \opt{r}_{i,\cdot|t},\W_{i,\cdot|t})\leq \Vmaxi$ ensures that at time $t+K$, the optimal cost of the tracking MPC fulfills 
\begin{equation} \label{eq:ref_v_max}
    \VNi(x_{i, t+K}, \opt{r}_{\subiplusk}, \mathbb{W}_{i, \subplusk}) \leq V_{\max, i}.
\end{equation}
\end{proposition}
% we derive an additional \textit{coupling} constraint to \eqref{eq:planner}
\begin{proof}
Given that $\X, \U$ are compact (Ass. \ref{assumption1}), and $\ell_i$ is quadratic as shown in \eqref{eq:bounds_on_stage_cost}, the cost $\JNi$ is Lipschitz continuous in $\rti$ with constant $\LipV > 0$. Given that the feasibility of \eqref{eq:tracking_feasible_trajectory} is independent of $\rti$, $\VNi$ is also Lipschitz with $\LipV$, i.e., for any two references $\rti, \hat{r}_{i, \sub}$ it holds that
%\begin{align}
%\label{eq:lipschitz_of_tracking_cost_wrt_ref}
%\begin{split}
%    | \VNi&(x_{i, t}, \rti, \mathbb{W}_{i,\sub}) - \VNi(x_{i, t}, \hat{r}_{i, \sub}, \mathbb{W}_{i,\sub}) | \\
%    &\leq | \JNi(\xti, \uti, \rti) - \JNi(\xti, \uti, \hat{r}_{i, \sub}) | \\
%    &\leq \LipV \mynorm{\rti[0:N_i] - \hat{r}_{i, \sub[0:N_i]}},
%\end{split}
%\end{align}
%
\begin{align}
\label{eq:lipschitz_of_tracking_cost_wrt_ref}
\begin{split}
\VNi&(x_{i, t}, \hat{r}_{i, \sub}, \mathbb{W}_{i,\sub}) \leq \JNi(\Tilde{x}_{i,\cdot|t}, \Tilde{u}_{i,\cdot|t}, \hat{r}_{i, \sub})\\
&\leq \JNi(\Tilde{x}_{i,\cdot|t}, \Tilde{u}_{i,\cdot|t}, r_{i, \sub}) + \LipV \mynorm{\rti[0:N_i] - \hat{r}_{i, \sub[0:N_i]}}\\
&= \VNi(x_{i, t}, \rti, \mathbb{W}_{i,\sub}) + \LipV \mynorm{\rti[0:N_i] - \hat{r}_{i, \sub[0:N_i]}},
\end{split}
\end{align}
where $(\Tilde{x}_{i,\cdot,t}, \Tilde{u}_{i,\cdot,t})$ is the optimal solution to Problem~\eqref{eq:tracking_feasible_trajectory} for the reference $\rti$. Given this reference, Theorem~\ref{theorem:stability_of_tracking_mpc} ensures the tracking cost decreases exponentially \eqref{eq:lyapunov_decrease_standard_tracking}, i.e., for any time $t+\delta$ with $\delta \in \mathbb{N}$ it holds that
%Given Theorem \ref{theorem:stability_of_tracking_mpc}, the tracking cost shows an exponential decrease over time for a given reference, hence \eqref{eq:lyapunov_decrease_standard_tracking} with Assumption~\ref{assumption:local_controllability} implies that the tracking cost at some time $t+\delta$ with $\delta \geq 0$ can be upper bounded by 
\begin{align}
\begin{split}
\label{eq:fast_coverage_exp_cost_decrease_K_steps}
    \VNi(x_{i, t+\delta},& \opt{r}_{i,\cdot + \delta|t}, \W_{i,\cdot + \delta|t}) \\
    &\leq \left( 1 - \frac{\alpha_{N,i}}{\gammabari} \right)^\delta \VNi(x_{i, t}, \opt{r}_{i,\cdot |t},\W_{i,\cdot|t}).
\end{split}
\end{align}
Finally, \eqref{eq:ref_v_max} follows by combining the previous equations:
\begin{align}
\label{eq:derivation_coupling_constraint}
\begin{split}
     & \quad \VNi(x_{i, t+K}, \opt{r}_{\subiplusk}, \mathbb{W}_{i, \subplusk}) \\
     & \quad \stackrel{\eqref{eq:lipschitz_of_tracking_cost_wrt_ref}}{\leq} \VNi(x_{i, t+K}, \opt{r}_{i,\cdot + K|t}, \W_{i,\cdot + K|t}) \\
     &\quad \quad + \LipV \mynorm{r_{\subi[K:K+N]} - r_{\subiplusk[0:N]}}\\
     & \quad \stackrel{\eqref{eq:fast_coverage_exp_cost_decrease_K_steps}}{\leq} \left( 1 - \frac{\alpha_{N,i}}{\gammabari} \right)^K \VNi(x_{i, t}, \opt{r}_{i,\cdot | t},\W_{i,\cdot|t}) \\
     &\quad \quad + \LipV \mynorm{\opt{r}_{\subi[K:K+N]} - \opt{r}_{\subiplusk[0:N]}}
     \stackrel{\eqref{eq:ref_v_max}}{\leq} V_{\max, i},
\end{split}
\end{align}
where the last inequality is enforced.
\end{proof}
This proposition ensures that any reference with the \textit{coupling} constraint~\eqref{eq:periodic_coupling_constraint_derivation} lies within the region of attraction of the tracking MPC. 
Hence, we enforce inequality~\eqref{eq:periodic_coupling_constraint_derivation} directly in the computation of a new reference $\opt{r}_{\subiplusk}$:
\begin{subequations} \label{eq:periodic_reference_planner}
\begin{align} 
    %\opt{r}_{\subiplusk}(\rti, \mathbb{W}_{i,\sub}, \VNit) \coloneqq \\
    \opt{r}_{\subiplusk} &\coloneqq
    \argmin_{r_{\subiplusk}} \quad H_i^T(p_{\subiplusk}, \W_{i,\sub[\cdot + K]}, t+K) \\
    \textrm{s.t.} &\quad \eqref{eq:planner1}, \eqref{eq:planner2}, \eqref{eq:planner3}, \eqref{eq:planner4}, \eqref{eq:planner5},\\
     &\quad \mynorm{\opt{r}_{\subi[K:K+N]} - \opt{r}_{\subiplusk[0:N]}} \leq \C(\VNit).\label{eq:periodic_coupling_constraint}
\end{align}
\end{subequations}

%To ensure convergence of the algorithm, $r_{i,\cdot | t + K} = \opt{r}_{i,\cdot + K|t}$ is used if the coverage cost is not strictly lower than that of the previous solution.

Similar to the approach in Problem~\eqref{eq:planner}, we select $\opt{r}_{i,\cdot | t + K} = \opt{r}_{i,\cdot + K|t}$ if the optimal coverage cost is not strictly lower than that of the previous optimal solution.

\subsection{Partition Update Verification}
\label{subsection:partition_update}
Once a new collection of reference trajectories is determined for each agent, we proceed to update the Voronoi partition sequence $\mathbb{W}_{\cdot|t+K}$, which will be utilized in the subsequent iteration at time $t + K$. However, it must be verified that a partition update is feasible for the tracking MPC, given its independence from the planner. This verification is important since the closed-loop trajectories from the tracker generally differ from the planned ones.

Thus, we consider the candidate partition sequence $\hat{\mathbb{W}}_{\cdot|t+K} := \mathbb{W}_{\opt{p}_{\cdot|t+K}}$ computed with (\ref{eq:voronoi_partition}), and define the following candidate trajectory, constructed from the optimal trajectory calculated by the tracking MPC at time $t$
\begin{equation}\label{eq:fast_coverage_candidate_sequence}
\begin{aligned}
    \hat{u}_{i, \cdot | t}   &= [ \optt{\utim[1]},\dots,\optt{\utim[N_i-2]},\: \opt{\utim[N_i-1]} , \opt{\utim[N_i]}], \\
    \hat{x}_{i, \cdot | t}   &= [ \optt{\xtim[1]},\dots, \optt{\xtim[N_i-2]} , \\
    & \quad \optt{\xtim[N_i-1]}, \hat{x}_{i, N_i-1 | t} , \hat{x}_{i, N_i | t}],
\end{aligned}
\end{equation}
where the last elements in the sequence are computed with the reference inputs $\opt{u}$ and we define $\hat{x}_{i, N_i-1 | t}, \hat{x}_{i, N_i | t}$ as 
\begin{align*}
\hat{x}_{i, N_i-1 | t} &= f(\optt{\xtim[N_i-1]}, \opt{\utim[N_i-1]}),\\
\hat{x}_{i, N_i | t} &= f(\hat{x}_{i, N_i-1 | t}, \opt{\utim[N_i]}). 
\end{align*}
Then, we verify whether the candidate reference trajectory of each agent fulfills the following conditions at time $t+K$
\begin{equation} \label{eq:partition_update_condition}
\begin{aligned}
     &\hat{x}_{\subiplusk[k]} \in \mathbb{X}_i, \\
     &C\hat{x}_{\subiplusk[k]}=:\hat{p}_{\subiplusk[k]} \in \hat{\Wbar}_{i,\subplusk[k]},\\
     &\forall k \in \{ 0, \dots, N_i \}.
\end{aligned}
\end{equation}

If the candidate sequences \eqref{eq:fast_coverage_candidate_sequence} fulfill the update conditions~\eqref{eq:partition_update_condition} for all agents, the new partition sequence is applied: $\mathbb{W}_{\cdot|t+K} = \hat{\mathbb{W}}_{\subplusk}$. Otherwise, the previously computed shifted partition sequence is applied, i.e., $\W_{\cdot|t+K} =\W_{\cdot+K|t}$.

\subsection{Theoretical Analysis}
In the following, we provide formal guarantees of Algorithm~\ref{alg:decoupled_coverage_control}. To this end, we introduce the following proposition, which provides a bound on the tracking MPC cost after performing a partition update.

\begin{proposition} \label{proposition:V_max_reference} Let Assumption~\ref{assumption1} hold. Given a partition sequence $\Wsub$ and reachable references $\opt{r}_{\subplusk}$ computed with \eqref{eq:periodic_reference_planner},
    suppose that $\VNi(x_{i,t}, \rti, \Wsubi) \leq \Vmaxi$ for all agents $i \in \M$ and $N_i > \max\{\opt{N}_i, N_{0,i}\}$ with $\opt{N}_i$ from \eqref{eq:n_star_lower_bound}, and let the candidate $\hat{\W}_{\subplusk} := \W_{\opt{p}_{\subplusk}}$ be the corresponding Voronoi partition sequence based on $\opt{r}_{\subplusk}$. If condition \eqref{eq:partition_update_condition} is fulfilled at time $t+K$, then it holds that 
    \begin{equation} \label{eq:V_max_reference_in_new_partitions}
        \VNi(x_{i, t+K}, \opt{r}_{\subiplusk}, \hat{\W}_{\subiplusk}) \leq V_{\max, i}. 
    \end{equation}
\end{proposition}
The proof can be found in Appendix~\ref{appendix:proof_V_max_reference}.
Further, the update conditions in \eqref{eq:partition_update_condition} are guaranteed to be fulfilled in finite time given the following proposition. 
\begin{proposition} \label{proposition:finite_time_partition_update}
%\textcolor{blue}{There exists a uniform bound $\tau$, such that for any $t\in\mathbb{N}$, there exists $t'\in[t,t+\tau]$, such that the reference is updated $r_{\cdot+1|t'}\neq r_{\cdot|t'+1}$ or the candidate trajectory $\hat{x}_{i, \cdot | t' + 1}$~\eqref{eq:fast_coverage_candidate_sequence} fulfill the conditions~\eqref{eq:partition_update_condition} at $t'$.}

     %There exists a uniform bound $\tau$ for any partition sequence $\Wsub$ and reachable references $\rt$ that are fixed, i.e. $r_{\cdot|t+k}=r_{\cdot+k|t} \, \forall k\in \mathbb{N}$, with $\VNi(x_{i,t}, \rti, \Wsubi) \leq \Vmaxi \, \forall i \in \M$, such that the candidate trajectories $\hat{x}_{i, \cdot | t' + 1}$~\eqref{eq:fast_coverage_candidate_sequence} constructed at time $t' \in \left[t, t + \tau \right]$ fulfill the conditions~\eqref{eq:partition_update_condition}.

    % There exists a uniform bound $\tau$, such that for any $t\in\mathbb{N}$, there exists a time $t'\in[t,t+\tau]$ at which the candidate trajectories $\hat{x}_{i, \cdot | t' + 1}$~\eqref{eq:fast_coverage_candidate_sequence} fulfill the conditions~\eqref{eq:partition_update_condition}, given that the reference is fixed, i.e. $r_{\cdot+k|t} = r_{\cdot|t + k} \, \forall k \in \mathbb{N}$, and $\VNi(x_{i,t}, \rti, \Wsubi) \leq \Vmaxi \, \forall i \in \M$.
     
Let Assumption~\ref{assumption1} and~\ref{assumption:local_controllability} hold. Suppose that after some $t\in\mathbb{N}$, the  
reference trajectories are reachable and fixed, i.e., $r_{i,\cdot|t + k} = r_{i,\cdot+k|t}$ and $\VNi(x_{i,t}, \rti, \Wsubi) \leq \Vmaxi \, \forall i \in \M$ $\forall k \in [0,\tau]$ for some uniform bound $\tau>0$. Then, there exists a time $t'\in[t,t+\tau]$ at which the candidate trajectories $\hat{x}_{i, \cdot | t' + 1}$~\eqref{eq:fast_coverage_candidate_sequence} satisfy the conditions~\eqref{eq:partition_update_condition}.
\end{proposition}
The proof of Proposition~\ref{proposition:finite_time_partition_update} can be found in Appendix~\ref{appendix:close_enough}. 
Subsequently, we can summarize the theoretical guarantees of Algorithm~\ref{alg:decoupled_coverage_control} in
the following theorem.
\begin{theorem} \label{theorem:decoupled_coverage_control}
    Let Assumptions~\ref{assumption1},~\ref{assumption:fintie_centroidal_periodic_configuration},  and~\ref{assumption:local_controllability} hold, and consider the initialization in Algorithm~\ref{alg:decoupled_coverage_control}. Then, for a horizon $N > \max \{N_{0,i}, \opt{N}\}$, it holds that

        \textit{I)} (Recursive feasibility) All the optimization problems in  Algorithm~\ref{alg:decoupled_coverage_control} are feasible for all $t\in\mathbb{N}$.
        
        \textit{II)} (Constraint satisfaction) The resulting closed-loop trajectories from Algorithm~\ref{alg:decoupled_coverage_control} satisfy the state and input constraints~\eqref{mpc_state_and_input_constraints} and ensure collision avoidance~\eqref{collision_avoidance}.
                
        \textit{III)} The update condition \eqref{eq:partition_update_condition} is fulfilled after a finite number of time steps.
        
        \textit{IV)} The closed-loop trajectories converge asymptotically to an optimal periodic configuration (Def. \ref{def:optimal_periodic_configuration}).
\end{theorem}

\begin{proof}\text{ }\\
    \textit{I)} The initial feasibility of the planner is ensured by the steady-state of all agents at $t=0$, which is given at initialization. Additionally, each planner iteration of \eqref{eq:periodic_reference_planner} remains recursively feasible with the candidate trajectory $r_{i,\cdot|t+K} = \opt{r}_{i,\cdot + K|t}$, which always lies within the updated partition sequence $\W_{\opt{p}_{\cdot + K|t}}$, as established by Lemma~\ref{lemma:feasible_partition_update}. 
    
    Similarly, the tracking MPC is feasible at $t=0$ since the vehicles start at a steady-state. Furthermore, the initial reference trajectory $\opt{r}_{\cdot|0}$ satisfies all constraints in \eqref{eq:tracking_feasible_trajectory}, and ensures that $\VNi(x_{i, 0}, \opt{r}_{i, \cdot|0}, \mathbb{W}_{i, \cdot|0}) \leq V_{\max, i}$. Every subsequent reference computation is reachable and adheres to \eqref{eq:periodic_coupling_constraint}, maintaining the bound $\VNi(x_{i, t}, \rti, \Wsubi) \leq V_{\max, i}$ as of Proposition~\ref{proposition:coupling_constraint}, ensuring that the agents stay within the region of attraction of the tracking MPC. Condition \eqref{eq:partition_update_condition} ensures that the partition update is only performed if feasible, and Proposition \ref{proposition:V_max_reference} guarantees the upper bound $\VNi(x_{i, t+K}, \opt{r}_{\subiplusk}, \W_{\opt{p}_{i,\subplusk}}) \leq V_{\max, i}$ on the tracking cost after any partition update. Therefore, the tracking MPC is recursively feasible by Theorem~\ref{theorem:stability_of_tracking_mpc}. 
    
    \textit{II)} The tracking MPC \eqref{eq:tracking_feasible_trajectory} enforces the state and input constraints~\eqref{mpc_state_and_input_constraints} at all times. Moreover, the constraints~\eqref{mpc_state_and_input_constraints} guarantee that each agent stays within its tightened partition sequence, thereby ensuring collision avoidance~\eqref{collision_avoidance}.

    \textit{III)} Suppose, for contradiction, that condition~\eqref{eq:partition_update_condition} is not satisfied $\forall t' \in \left [ t, \Tilde{t} \right]$ with some finite $t, \Tilde{t} \in \mathbb{R}_{+}$, and the partitions are not updated, i.e., they are only shifted at each time step with \eqref{eq:shifed_partition_periodic}.
    
    We first show that the references converge to a fixed trajectory. Since it holds that $\C(V_{i,t}) \geq \C(\Vmaxi) > 0$ with $V_{i,t} \leq \Vmaxi$ from \textit{Part I}, the planner cost is non-increasing with each iteration of \eqref{eq:periodic_reference_planner}, which is lower-bounded by zero. Further, by convention $r^\star_{i,\cdot | t + K} = \opt{r}_{i,\cdot + K|t}$ is used if the coverage cost remains constant with the minimizer. Given the compact constraints in \eqref{eq:periodic_reference_planner}, the sequence of reference updates $\opt{r}_{i, \cdot | t + n K}$ converges for some finite $n \in \mathbb{N}$. Hence, there exists a finite $\hat{t} \geq t$ at which the reference trajectories are not updated, i.e., $\opt{r}_{\cdot | \hat{t} + K} = \opt{r}_{\cdot + K | \hat{t}}$. Note that Problem~\eqref{eq:periodic_reference_planner} is equivalent to the planning Problem~\eqref{eq:planner} with an additional trust-region constraint on the maximal change of the reference~\eqref{eq:periodic_coupling_constraint}. Consequently, constraint \eqref{eq:periodic_coupling_constraint} becomes inactive at the convergence of the reference, and the solution of \eqref{eq:periodic_reference_planner} equivalently represents a local minimizer of \eqref{eq:planner}. For such fixed references, Proposition~\ref{proposition:finite_time_partition_update} guarantees a finite time $\Tilde{t} \in \left[ \hat{t}, \hat{t} + \tau \right]$ such that the closed-loop trajectories converge close enough (Appendix~\ref{appendix:close_enough}) to these references, fulfilling \eqref{eq:partition_update_condition}. This contradicts the initial assumption and hence condition \eqref{eq:partition_update_condition} is met at time $t' \in \left [ t, \Tilde{t} \right]$.

    \textit{IV)} This part follows the steps in the proof of Proposition~\ref{proposition:planner_convergence}. \textit{Part III} shows that the partition update condition \eqref{eq:partition_update_condition} is satisfied for some $t' \in \left [t,\hat{t} + \tau \right ]$. Therefore, the coverage cost using Algorithm~\ref{alg:decoupled_coverage_control} decreases with each iteration, forming a descent function:
    \begin{subequations}
    \begin{align}
    \begin{split} \label{eq:descent_3}
        & H^T(\opt{p}_{\cdot|t}, \W_{\opt{p}_{\cdot|t}}, t)\\
        &\geq H^T(\opt{p}_{\cdot|t+1}, \W_{\opt{p}_{\cdot + 1|t}}, t + 1)\\
        &\geq \dots \geq H^T(\opt{p}_{\cdot|t'}, \W_{\opt{p}_{\cdot + t' - t|t}}, t')\\
    \end{split}\\
    \begin{split} \label{eq:descent_4}
        &\geq H^T(\opt{p}_{\cdot|t'}, \W_{\opt{p}_{\cdot|t'}}, t')\\
        &\geq 0, \forall \, t \in \mathbb{N}.
    \end{split}
    \end{align}
    \end{subequations}
    The first inequalities~\eqref{eq:descent_3} hold from the cost decrease and convergence of the planner \eqref{eq:planner} from \textit{Part III}. The last two inequalities~\eqref{eq:descent_4} follow the argumentation of the proof of Proposition~\ref{proposition:planner_convergence}. 
    Therefore, Algorithm~\ref{alg:decoupled_coverage_control} has a descent function and converges \cite[Prop. 3.3]{Corts2002CoverageCF}.
    
    %Algorithm~\ref{alg:decoupled_coverage_control} converges to an optimal periodic configuration (Def. \ref{def:optimal_periodic_configuration}) of \eqref{eq:periodic_reference_planner} . 
    
    Assume, for contradiction, that the configuration that the algorithm converges to is not a minimizer of \eqref{eq:planner}. Convergence to any periodic configuration implies convergence of the partition sequence $\lim_{t\to\infty}  \W_{\opt{p}_{\cdot,t}}$. From \textit{Part III}, the solution of \eqref{eq:periodic_reference_planner} equivalently represents a local minimizer of \eqref{eq:planner} at convergence. Consequently, \eqref{eq:periodic_reference_planner} converges to a minimizer of \eqref{eq:planner} for a fixed periodic partition sequence, contradicting the assumption. Therefore, Algorithm~\ref{alg:decoupled_coverage_control} converges to an optimal periodic configuration of problem \eqref{eq:planner} (Def. \ref{def:optimal_periodic_configuration}).
    
    With the convergence of the reference trajectories $\lim_{t\to\infty} r^\star_{\cdot|t}$, the closed-loop systems exhibit exponential convergence to these trajectories, as established by Theorem~\ref{theorem:stability_of_tracking_mpc}.
\end{proof}

It follows from Theorem~\ref{theorem:decoupled_coverage_control} that the closed-loop trajectories of Algorithm~\ref{alg:decoupled_coverage_control} converge to an optimal periodic configuration of \eqref{eq:planner} while satisfying all constraints and ensuring collision avoidance.

\section{Nonperiodic coverage MPC}
\label{section:nonperiodic_coverage_control_mpc}
This section presents a practical approach to coverage control for nonperiodic density functions. The proposed method introduces targeted modifications to Algorithm \ref{alg:decoupled_coverage_control}. We start by presenting the following minimization problem solved by the planner at time $t$.
\begin{subequations} \label{eq:nonperiodic_reference_planner}
\begin{align} 
    %\hat{r}_{\subiplusk}(\rti, \mathbb{W}_{i,\sub}, \VNit) \coloneqq \\
    \opt{r}_{\subiplusk} \coloneqq
    \argmin_{r_{\subiplusk}} &\quad H_i^T(p_{\subiplusk}^r, \W_{i,\sub[\cdot + K]}, t+K) \\
    \textrm{s.t.} \quad &(\ref{eq:planner1}), (\ref{eq:planner3}), (\ref{eq:planner4}), (\ref{eq:planner5}),\\
    &x_{i, T - 1 \mid t} = f_i(x_{i,T - 1 \mid t}, u_{i,T - 1 \mid t}), \label{eq:nonperiodic_steady_state_constraint} \\
     &r_{\subi[K:K+N_i]} = r_{\subiplusk[0:N_i]}.\label{eq:nonperiodic_coupling_constraint}
\end{align}
\end{subequations}
Notably, we consider a horizon $T$ which is no longer associated with the period of the density function. To accommodate the nonperiodic density function, we remove the periodic constraint on the reference in \eqref{eq:planner3}. Instead, we enforce a steady-state constraint on the last element of the reference, which will ensure recursive feasibility. Therefore, reference trajectories and partition sequences are defined beyond the horizon $T$, which is necessary for computing a new partition sequence for the next iteration at time $t + K$. Therefore, we define partitions and trajectories shifted by $n \in \mathbb{N}$ steps in time for the nonperiodic algorithm as
\begin{subequations}
\begin{align}
    \label{eq:shifed_partition_nonperiodic}
     \mathbb{W}_{\cdot + n|t} &\coloneqq \{ \mathbb{W}_{\sub[n]}, \dots , 
     \mathbb{W}_{\sub[T-1]}, \dots,  \mathbb{W}_{\sub[T-1]} \},\\
     \label{eq:shifted_reference_nonperiodic}
     r_{\cdot + n|t} &\coloneqq \{ r_{n|t}, \dots,
        r_{T-1|t}, \dots, r_{T-1|t}  \},
\end{align}
\end{subequations}
where the last $n$ elements in the sequence can be repeated since they satisfy the dynamics constraints.
 
Consequently, we introduce a new scheme which extends Algorithm \ref{alg:decoupled_coverage_control} for nonperiodic density functions.
\begin{algorithm}
\caption{Nonperiodic Coverage Control MPC}
\label{alg:nonperiodic_coverage_control}
  \SetAlgoLined
    \textbf{Algorithm}~\ref{alg:decoupled_coverage_control} with Problem~\eqref{eq:periodic_reference_planner} replaced by Problem~\eqref{eq:nonperiodic_reference_planner}\\
        Shifting operations \eqref{eq:shifted_reference_periodic}-\eqref{eq:shifed_partition_periodic} replaced by \eqref{eq:shifed_partition_nonperiodic}-\eqref{eq:shifted_reference_nonperiodic}
\end{algorithm}

The properties of the previous algorithm are summarized in the following theorem.

\begin{theorem} \label{theorem:decoupled_coverage_control_nonperiodic}
    Let Assumptions~\ref{assumption1} and~\ref{assumption:local_controllability} hold. Then, for a horizon $N_i \geq \max \{N_{0,i}, \opt{N}_i\}$, it holds that

    \textit{I)} (Recursive feasibility) All the optimization problems in  Algorithm~\ref{alg:nonperiodic_coverage_control} are feasible for all $t\in\mathbb{N}$.
    
    \textit{II)} (Constraint satisfaction) The resulting closed-loop trajectories from Algorithm~\ref{alg:nonperiodic_coverage_control} satisfy the state and input constraints~\eqref{mpc_state_and_input_constraints} and ensure collision avoidance \eqref{collision_avoidance}.

    \textit{III)} The update condition \eqref{eq:partition_update_condition} is fulfilled after a finite number of time steps.
        
\end{theorem}

\begin{proof} $\quad$\\
    \textit{I)} The optimization problem of the planner \eqref{eq:nonperiodic_reference_planner} is always feasible since the shifted reference $\rti[\cdot + K]$ with \eqref{eq:shifted_reference_nonperiodic} is a feasible candidate for problem \eqref{eq:periodic_reference_planner} at time $t+K$ due to the incorporation of the steady-state. 
    The recursive feasibility of the tracking MPC follows analogously to Theorem~\ref{theorem:decoupled_coverage_control}.\\
    \textit{II)} The argumentation follows the proof of Theorem~\ref{theorem:decoupled_coverage_control}.\\
    \textit{III)} We leverage Proposition~\ref{proposition:finite_time_partition_update} to show the finite-time partition update condition for Algorithm~\ref{alg:nonperiodic_coverage_control}. Although this proposition requires a fixed reference trajectory, we show that this is still valid for a reference trajectory that is updated over time with the consistency constraint in \eqref{eq:nonperiodic_coupling_constraint}. 
    In particular, consider $r = \left [ r_{0:K|t}, r_{0:K|t + K}, \ldots,  r_{0:K|t + nK}\right ] \forall n \in \mathbb{N}$. 
    This reference trajectory is reachable, and the operation of the tracking MPC depends solely on this part of the reference, which can be treated as fixed. Thus, the arguments in Proposition~\ref{proposition:finite_time_partition_update} leveraging the convergence of the tracking MPC remain valid, i.e., the update condition \eqref{eq:partition_update_condition} are satisfied after at most $\tau$ steps. 
%    For instance, consider  the reference trajectory $r = \left [ r_{0:K|t}, r_{0:K|t + K}, \ldots,  r_{0:K|t + nK}\right ] \forall n \in \mathbb{N}$, which is constructed as a concatenation of all references that are applied to the tracking MPC until a partition update occurs. It follows that $r$ is not modified and therefore fixed from the tracking MPC perspective since it only observes the first $N_i + K$ elements of the reference at each time $t$. Additionally, $r$ is still reachable and satisfies the dynamics constraints \eqref{eq:planner1} due to the equality constraint \eqref{eq:nonperiodic_coupling_constraint}. Hence, Proposition~\ref{proposition:finite_time_partition_update} ensures a partition update at time $t' \in \left[t, t + \tau \right]$ for some finite and uniform bound $\tau$.
\end{proof}
\begin{comment}
%%Not possible to proof; since cost could actually increase; cost at steady-state doesn't, but that's not all we optimize...    
While the nonperiodic nature of the problem does not allow us to show convergence of the closed-loop trajectories from Algorithm~\ref{alg:nonperiodic_coverage_control}, we provide the following corollary which applies to any static density function.

\begin{corollary}
Suppose the conditions in Theorem~\ref{theorem:decoupled_coverage_control_nonperiodic} hold and the density is static, i.e., $\phi(q,t) = \phi(q,t+1) \, \forall t \in \mathbb{N},\ \forall q \in \mathbb{A}$. Then, the closed-loop trajectories of the agents in Algorithm~\ref{alg:nonperiodic_coverage_control} converge to a steady-state configuration.
\end{corollary}

\begin{proof}
    Given the time-invariant nature of the density function, any solution to the minimization problem in \eqref{eq:nonperiodic_reference_planner} serves as an upper bound on the coverage cost at each iteration. By Proposition~\ref{proposition:planner_convergence}, the agent trajectories converge. Suppose, for contradiction, that the agents do not converge to a steady-state configuration. This would imply persistent variations in agent trajectories between time steps, contradicting the convergence result of Proposition~\ref{proposition:planner_convergence}. Therefore, the agents converge to a steady-state configuration.
\end{proof}
\end{comment}

\section{Discussion}
\label{section:discussion}
In the following, we discuss the properties of the proposed periodic and nonperiodic coverage control algorithms and compare them with existing schemes for static density functions.

\subsection{Periodic Coverage Control}
To address the joint problem of reference and partition computation, we first extended Lloyd's algorithm \cite{lloyd} to periodic dynamic problems. More specifically, we solve for reachable and periodic trajectories that minimize the coverage cost for the entire period $T$ and compute the Voronoi partition for each element in the trajectory sequence. Each iteration results in a reduced coverage cost through both trajectory optimization and partition update. This methodology differs from existing coverage control algorithms for static densities, where a constant partition is used in the MPC scheme \cite{rickenbach2023}. Furthermore, we introduced the concept of optimal periodic configurations, which represent a locally optimal solution of partitions and trajectories for a periodic density, as a natural extension of the steady-state problem considered in~\cite{rickenbach2023} for the static coverage problem.

We address this problem by developing a two-layer framework composed of a reference planner and a tracking MPC. 
Similar to \cite{rickenbach2023}, we use a tracking MPC without terminal ingredients, which facilitates implementability. Furthermore, we employ a partition update verification to ensure recursive feasibility, which is also inspired by this work.
However, our framework departs from \cite{rickenbach2023} in several key aspects. 
Most notably, the approach in~\cite{rickenbach2023} ensures recursive feasibility by using artificial references in the tracking MPC. This increases computational complexity, especially in the considered problem of periodic references.  In contrast, the proposed planner directly computes reachable references that minimize the coverage cost while accounting for the tracker’s region of attraction, the active partition sequence, and the latest information on the density function. 
Furthermore, we account for the computational complexity of the planner by using a multi-rate implementation, which requires updates of the reference trajectory only every $K$ time steps.
Thus, we can implement a simple trajectory tracking MPC that maintains a bound on the tracking cost and ensures closed-loop recursive feasibility, even under updates of the reference trajectory.  
To conclude, our proposed framework converges to an optimal periodic configuration with reduced computational complexity through the proposed architecture.

We note that there exists a trade-off between the convergence rate of the periodic algorithm and the size of the region of attraction of the tracking MPC, expressed by the bounds $\Vmaxi$. The convergence speed of the algorithm is also related to the execution rate of the planner, which is limited by the computational resources of the vehicle. Meanwhile, the tracking MPC is computationally efficient and can be executed at high rates, which is crucial for reliable trajectory tracking. 

Finally, we highlight that our method can be easily extended to different stage cost functions as in \cite{rickenbach2023}, which might be more suitable for non-holonomic systems or different applications. Overall, Algorithm \ref{alg:decoupled_coverage_control} provides a practical two-layer implementation for time-varying coverage control with robust convergence, constraint satisfaction, and recursive feasibility guarantees for periodic density functions.

\subsection{Nonperiodic Coverage Control}
To solve the problem of nonperiodic coverage control, we propose Algorithm~\ref{alg:nonperiodic_coverage_control}, which adapts the previous scheme by replacing the periodic constraint with a final steady-state and a more restrictive equality constraint on the computation of new references. With these modifications, we guarantee recursive feasibility, constraint satisfaction, and finite-time partition updates. The proposed scheme balances trajectory optimality with computational complexity, both of which depend on the execution rate of the planner. We further demonstrate the effective performance for both periodic and nonperiodic densities in Section~\ref{section:results}, achieving rapid coverage cost reduction. This makes the algorithm well-suited for resource-limited and nonperiodic environments.

\section{Experimental Results}
\label{section:results}
This section describes the experimental setup and results that demonstrate our proposed methods.

\subsection{Experimental Setup}

The following experiments were performed on a fleet of $M=4$ miniature RC cars, scaled at 1:28, in combination with CRS, a software framework intended for both single and multi-agent robotics and control \cite{carron2023chronos}. We employ ROS \cite{ros} to facilitate communication between tracker nodes, planner nodes, and agents. The planned trajectories from \eqref{eq:planner} are optimized using \textit{IPOPT} \cite{wachter2006implementation}, while the tracking MPC problem \eqref{eq:tracking_feasible_trajectory} is solved using \textit{Acados}~\cite{verschueren2022acados}. Additionally, the vehicle's state is measured using the motion capture system from Qualiysis. The nonlinear dynamics of the vehicle are approximated with a kinematic bicycle model \cite{rajamani2011vehicle}. The continuous-time dynamics are represented as
\begin{align*}
\dot{x} &= 
\begin{bmatrix}
    \dot{p}_{\mathrm{x}} \\
    \dot{p}_{\mathrm{y}} \\
    \dot{\psi} \\
    \dot{v}
\end{bmatrix}
=
\begin{bmatrix}
    v \cos(\psi + \beta) \\
    v \sin(\psi + \beta) \\
    \frac{v}{l_{\mathrm{r}}} \sin(\beta) \\
    a
\end{bmatrix},\\
\beta &= \arctan\left(\frac{l_{\mathrm{r}}}{l_{\mathrm{r}} + l_{\mathrm{f}}} \tan(\delta)\right),
\end{align*}
where the state and input vectors are defined as 
\begin{align*}x = \begin{bmatrix} p_{\mathrm{x}} & p_{\mathrm{y}} & \psi & v \end{bmatrix}^\top,
 u = \begin{bmatrix} \delta & a \end{bmatrix}^\top.
\end{align*}
In this context, $p_{\mathrm{x}/\mathrm{y}}$ refers to the position of the vehicle, $\psi$ denotes its heading angle, $v$ represents the velocity, $\beta$ indicates the slip, $\delta$ represents the steering angle, and $a$ describes the acceleration. The function $f$ is derived through the application of a Runge-Kutta 4th-order discretization method at $33$ ms. 

Note that the known and Lipschitz continuous dynamics required by Assumption \ref{assumption1} are satisfied by the kinematic bicycle model, and we employ compact constraint sets for both state and input vectors. Furthermore, Assumption \ref{assumption:local_controllability} holds if the linearized dynamics are stabilizable, which is valid for the vehicle model assumed that a minimal positive velocity can be ensured.

The density function $\phi$ is characterized as a two-dimensional multivariate normal distribution with a time-varying mean $(\mu_x(t), \mu_y(t)) \in \mathbb{A}$
\begin{align*}
        \phi(q,t) &= e^{- \frac{1}{2 \sigma^2} \left( (q_x - \mu_x(t))^2 + (q_y - \mu_y(t))^2 \right)}.
\end{align*}
The selected area of operation $\mathbb{A}$ has dimensions $[-2,2] \times [-2,2]$ in $\mathrm{m}$. The quadratic cost is given by $Q = \text{diag}(180, 180, 1, 1)$ and $R = \text{diag}(0.1,0.1)$. Further, we select $\LipV = 180$, $1 - \frac{\alpha_{N,i}}{\gammabari} = 0.95$, $V_{\max,i} = 70$, $\epsilon = 0.005$ and $R_{\max} = 0.055$ m.

The overall implementation considers a prediction horizon of $N_i = 20$ and a sampling time of $33 \, \mathrm{ms}$ for the tracking MPC. The planner is executed every $K = 190$ steps for the periodic planner with $T=150$, and $K \in \{30,60\}$ for the nonperiodic planner with a horizon $T = 100$.

\subsection{Results for Periodic Densities}

\begin{figure}[!h]
    \centering
    \begin{subfigure}{1.0\columnwidth}
        \includegraphics[width=\textwidth]{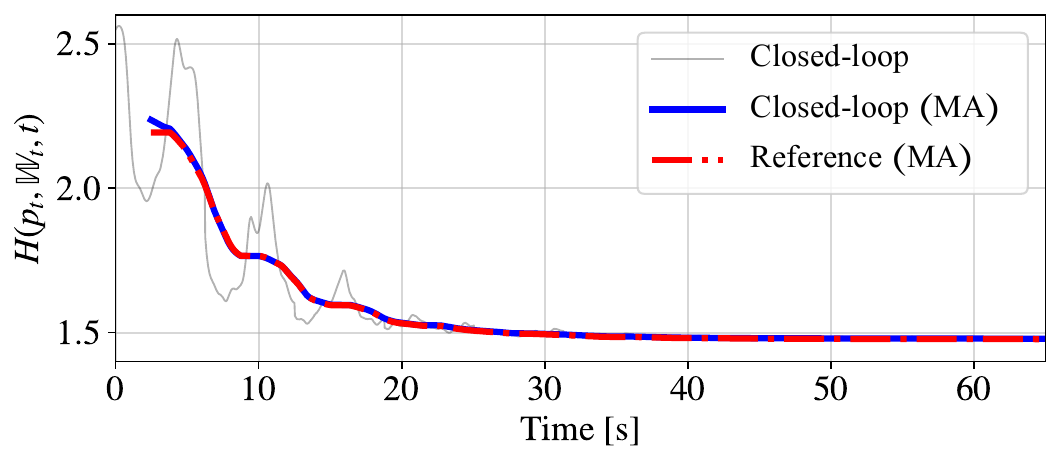}
        \caption{Closed-loop coverage cost over time. The moving average (MA) of the coverage cost is presented for both the reference trajectories and the closed-loop trajectories.}
         \hfill
        \label{fig:coverage_cost_periodic}
    \end{subfigure}

    \begin{subfigure}{1.0\columnwidth}
        \includegraphics[width=\textwidth]{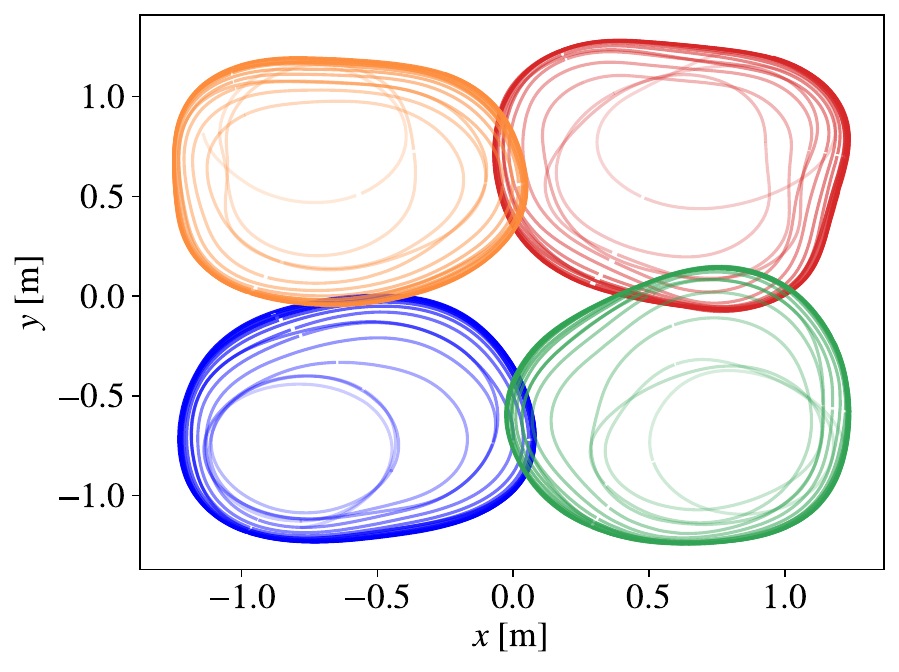}
        \caption{Closed-loop trajectories of the four agents, shown in different colors.}
        \label{fig:closed_loop_periodic}
    \end{subfigure}
    \caption{Performance of the periodic coverage MPC (Algorithm~\ref{alg:decoupled_coverage_control}) for a periodic density function. (a) Coverage cost over time, and (b) agent trajectories.}
\end{figure}

Algorithm \ref{alg:decoupled_coverage_control}, which runs the periodic planner, is evaluated to demonstrate the evolution of the coverage cost over time. The mean of the periodic density follows a circular motion of radius $r = 0.9 \, \mathrm{m}$ with a period of $4.95$s. Figure \ref{fig:coverage_cost_periodic} illustrates the coverage cost corresponding to each time step $t$, i.e., $H(p_t, \W_t, t)$, and its moving average of window $T$, which is proportional to the closed-loop coverage cost $H^T$. It can be observed that the average coverage cost decreases over time, achieving convergence for both reference and closed-loop trajectories. The closed-loop trajectories of the four agents are depicted in Figure~\ref{fig:closed_loop_periodic}. It can be seen that these trajectories converge to an optimal periodic configuration as delineated in Definition~\ref{def:optimal_periodic_configuration}. It should be noted that, despite the paths overlapping in the figure, the minimum distance between agents is always ensured. Additionally, the partition update conditions are satisfied in every iteration.

\begin{figure}[!h]
    \centering
    \begin{subfigure}{1.0\columnwidth}
        \includegraphics[width=\textwidth]{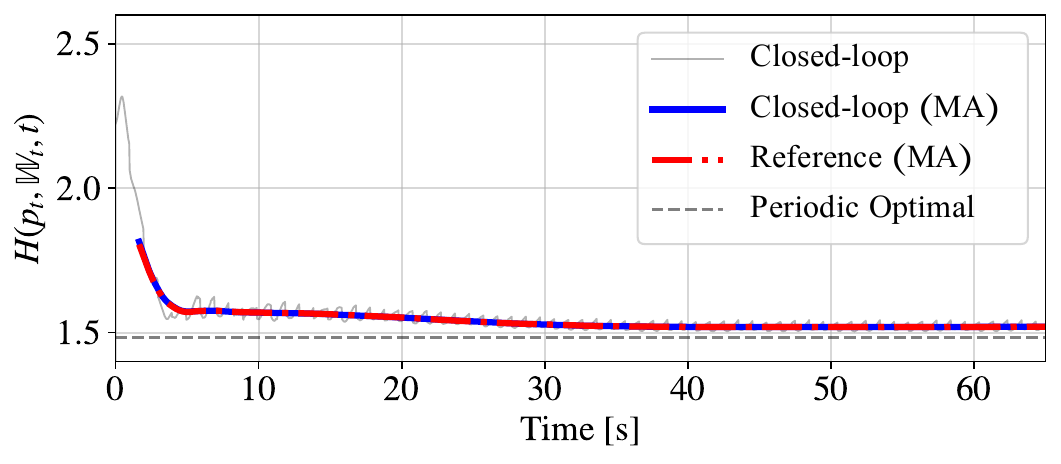}
        \caption{Closed-loop coverage cost over time. The moving average (MA) of the reference and closed-loop coverage cost is presented. Additionally, the optimal closed-loop coverage cost achieved at convergence with the periodic planner is presented.}
        \label{fig:coverage_cost_periodic_steady_state}
    \end{subfigure}

    \begin{subfigure}{1.0\columnwidth}
        \includegraphics[width=\textwidth]{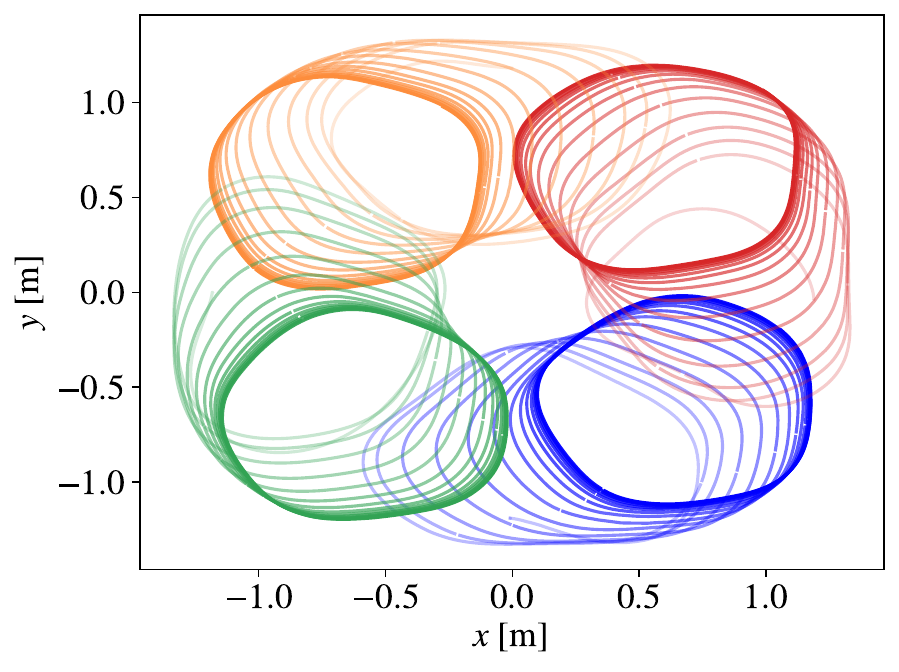}
        \caption{Closed-loop trajectories of the four agents, shown in different colors.}
        \label{fig:closed_loop_periodic_steady_state}
    \end{subfigure}
    
    \caption{Performance of the nonperiodic coverage MPC (Algorithm~\ref{alg:nonperiodic_coverage_control}) for a periodic density function. (a) Coverage cost over time, and (b) agent trajectories.}
\end{figure}

The results for the nonperiodic planner (Algorithm~\ref{alg:nonperiodic_coverage_control}) applied to the same periodic density are further presented. As illustrated in Figure~\ref{fig:coverage_cost_periodic_steady_state}, the coverage cost decreases faster, which can be attributed to the absence of a periodic constraint in the trajectory planning and a shorter horizon $T$. Subsequently, the coverage cost converges to a value slightly above the cost attained at convergence by Algorithm~\ref{alg:decoupled_coverage_control} due to the steady-state constraint \eqref{eq:nonperiodic_steady_state_constraint}. Additionally, the closed-loop trajectories are depicted in Figure~\ref{fig:closed_loop_periodic_steady_state}, showing a configuration analogous to that observed with the periodic planner.

\subsection{Results for Nonperiodic Densities}

This section provides the results of the nonperiodic Algorithm~\ref{alg:nonperiodic_coverage_control} as it addresses a nonperiodic density function whose mean follows an arbitrary path along the space $\mathbb{A}$.

\begin{figure}[!h]
    \centering
    \begin{subfigure}{1.0\columnwidth}
        \includegraphics[width=\textwidth]{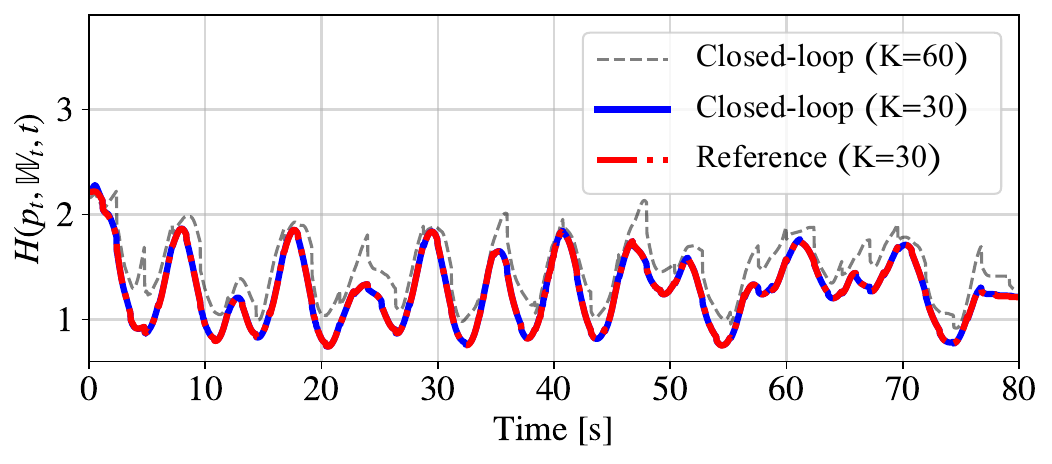}
        \caption{Closed-loop coverage cost over time with $K=30$ and $K=60$. The optimal cost of the planner is also presented for $K=30$.}
        \label{fig:coverage_cost_nonperiodic}
    \end{subfigure}

    \begin{subfigure}{1.0\columnwidth}
        \includegraphics[width=\textwidth]{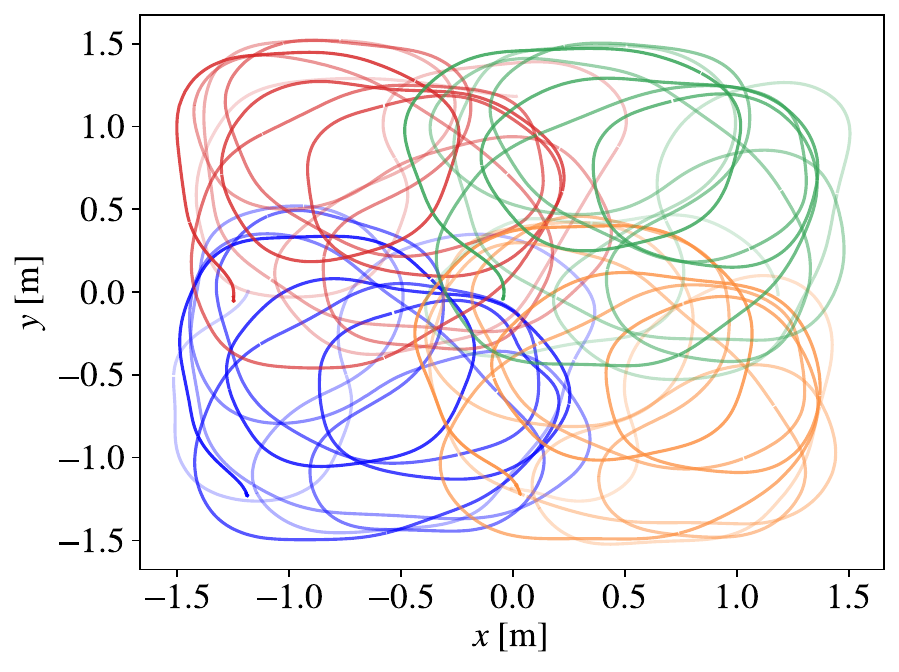}
        \caption{Closed-loop trajectories of the four agents for $K=30$, shown in different colors.}
        \label{fig:closed_loop_nonperiodic}
    \end{subfigure}
    
    \caption{Performance of the nonperiodic coverage MPC (Algorithm~\ref{alg:nonperiodic_coverage_control}) for a nonperiodic density function. (a) Coverage cost over time for different values of $K$, and (b) closed-loop agent trajectories.}
\end{figure}

Figure~\ref{fig:coverage_cost_nonperiodic} shows the temporal progression of the coverage cost associated with the Algorithm~\ref{alg:nonperiodic_coverage_control}. The coverage cost for $K=30$ is presented for both the planner references and tracker closed-loop, demonstrating a lower value compared to the closed-loop coverage cost observed for $K=60$. This hilights the trade-off between the algorithm's optimality and the frequency at which new references are generated. It is important to acknowledge that the coverage cost does not converge, given the nonperiodic nature of the density function.

\subsection{Computational Complexity}
\begin{figure}[!h]
\centering
\small
\begin{tabular}{|l|r|r|r|}
\hline
\textbf{Experiment} & \textbf{Mean (ms)} & \textbf{Std (ms)} & \textbf{Max (ms)} \\ \hline
Periodic (Alg.\ref{alg:decoupled_coverage_control})      & 944.24 & 328.43 & 2030.00 \\ \hline
Periodic (Alg.\ref{alg:nonperiodic_coverage_control})  & 107.00 & 70.41 & 453.78 \\ \hline
Nonperiodic (Alg.\ref{alg:nonperiodic_coverage_control}) & 96.69  & 61.34  & 490.83  \\ \hline
\end{tabular}
\caption{Solve times of IPOPT for different experiments (in milliseconds).}
\label{tab:ipopt_times}
\end{figure}

Figure~\ref{tab:ipopt_times} presents a comparison of planner solve times for each experiment provided by the \textit{IPOPT} numerical solver. It is observed that the nonperiodic  Algorithm~\ref{alg:nonperiodic_coverage_control} attains a lower average solve time compared to the periodic Algorithm~\ref{alg:decoupled_coverage_control}. This phenomenon is attributed to the horizon length $T$ employed, which can be shorter than the density period for Algorithm~\ref{alg:nonperiodic_coverage_control}. Furthermore, the periodic constraint in \eqref{eq:planner2} increases the complexity of the problem, resulting in increased solve times. Lastly, the solve times associated with Algorithm~\ref{alg:nonperiodic_coverage_control} are independent of the periodic or nonperiodic nature of the density function. Solve times for the tracking MPC in ACADOS consistently remain below $10$ ms with an average of $0.418$ ms, suggesting a significantly reduced computational complexity compared to the planning phase. Finally, the solve times respect the planner and tracker rates for both algorithms.

\section{Conclusion}
\label{section:conclusion}
The presented framework addresses the challenges of time-varying coverage control in nonlinear constrained dynamic systems by developing a two-layer planner-tracker MPC scheme. This framework integrates optimal trajectory planning with a tracking MPC without terminal ingredients. We rigorously establish guarantees on recursive feasibility, adherence to state and input constraints, and collision avoidance. Furthermore, we ensure convergence to an optimal periodic configuration for periodic density functions. The effectiveness of our approach is successfully validated in hardware experiments on a fleet of four small-scale vehicles.

\bibliography{bibliography} 
\bibliographystyle{ieeetr}

\appendix
\subsection{Proof of Proposition \ref{proposition:V_max_reference}}
\label{appendix:proof_V_max_reference}
We show that Proposition \ref{proposition:V_max_reference} guarantees the bound $\Vmaxi$ on the tracking cost in \eqref{eq:fast_coverage_exp_cost_decrease_K_steps} after applying a partition update.
\begin{proof}
Consider a reachable reference $\opt{r}_{\subiplusk}$ satisfying constraints in \eqref{eq:periodic_reference_planner} and the candidate partition sequence $\hat{\mathbb{W}}_{\subplusk} := \W_{\opt{p}_{\subplusk}}$. The update condition~\eqref{eq:partition_update_condition} ensures that the candidate sequence $(\hat{x}_{\subiplusk}, \hat{u}_{\subiplusk})$ from \eqref{eq:fast_coverage_candidate_sequence} is a feasible solution to problem~\eqref{eq:tracking_feasible_trajectory} with the updated partition sequence, i.e., $\hat{\mathbb{W}}_{i,\subplusk}$. Further, the following relation holds
\begin{align} \label{eq:bound_on_value_function_with_new_partition}
    \begin{split}
        \VNi&(x_{i, t+K}, \opt{r}_{\subi[\cdot + K]}, \hat{\mathbb{W}}_{i,\subplusk})\\
        &\leq \;  \JNi(\hat{x}_{\subiplusk}, \hat{u}_{\subiplusk},  \opt{r}_{\subi[\cdot + K]}) \\
        & \stackrel{\eqref{eq:fast_coverage_candidate_sequence}}{=} \sum_{k=1}^{N_i-2} \ell_i(\optt{x_{i, k | t+K-1}}, \optt{u_{i, k | t+K-1}},  \opt{r}_{\subi[K + k - 1]}) \\
        & \quad + \ell_i(\optt{x_{i, N_i-1 | t+K-1}}, \opt{u_{i, K+N_i-2| t}},  \opt{r}_{\subi[K+N_i - 2]}) \\ 
        & \quad + \ell_i(\hat{x}_{i, N_i-1 | t+K}, \opt{u_{i, K+N_i-1 | t}},  \opt{r}_{\subi[K+N_i - 1]})  \\
        & \leq \VNi(x_{i, t+K-1},  \opt{r}_{\subi[\cdot + K - 1]}, \mathbb{W}_{i,\cdot + K - 1 | t}) \\
        & \quad - \ell_i(\optt{x_{i, 0 | t+K-1}}, \optt{u_{i, 0 | t+K-1}},  \opt{r}_{\subi[K - 1]}) \\
        & \quad + \ell_i(\hat{x}_{i, N_i-1 | t+K}, \opt{u_{i, K+N_i-1 | t}},  \opt{r}_{i, K+N_i-1 | t}). \\
    \end{split}
\end{align}

We bound the last term in \eqref{eq:bound_on_value_function_with_new_partition} using the bounds on the stage cost~\eqref{eq:bounds_on_stage_cost} and Lipschitz continuity of the dynamics (Asm.~\ref{assumption1}):
    \begin{align*}
    \begin{split}
        &\ell_i(\hat{x}_{i, N_i-1 | t+K}, \opt{u_{i, K+N_i-1 | t}},  \opt{r}_{i, K+N_i-1 | t})\\
        & \quad \stackrel{\eqref{eq:bounds_on_stage_cost}}{\leq} \alpha_{2,i} \mynorm{\hat{x}_{i, N_i-1 | t+K} - \opt{x_{i, K+N_i-1 | t}}}^2 \\ 
        &  \quad \stackrel{\mathrm{Asm.} \ref{assumption1}}{\leq} \alpha_{2,i} \mathcal{L}_{f_i}^2 \mynorm{\optt{x_{i, N_i-1 | t+K-1}} - \opt{x_{i, K+N_i-2 | t}}}^2 \\ 
        &  \quad \stackrel{\eqref{eq:bounds_on_stage_cost}}{\leq} \frac{\alpha_{2,i}}{\alpha_{1,i}} \mathcal{L}_{f_i}^2 \ell_i(\optt{x_{i, N_i-1 | t+K-1}}, \opt{u_{i, K+N_i-2| t}},  \opt{r}_{i, K+N_i-2 | t}) \\
        &  \quad \leq \frac{\alpha_{2,i}}{\alpha_{1,i}} \mathcal{L}_{f_i}^2 \gammabari \left( \frac{\gammabari - 1}{\gammabari} \right)^{N_i} \\
        & \quad  \quad \times\ell_i(\optt{x_{i, 0 | t+K-1}}, \optt{u_{i, 0 | t+K-1}},  \opt{r}_{\subi[K - 1]}),
        % & \leq \frac{\alpha_{2,i}}{\alpha_{1,i}} \mathcal{L}_{f_i}^2 \VNi(x_{i, t+K-1},  \opt{r}_{i, \cdot | t+K-1}, \mathbb{W}_{\cdot | t+K-1, i}),
    \end{split}
\end{align*}

    where the last equation follows from \cite[Eq. (20)-(21)]{koehlerECO}.
    Applying this bound in \eqref{eq:bound_on_value_function_with_new_partition} yields
    \begin{align*}
        \begin{split}
            \VNi(&x_{i, t+K},  \opt{r}_{\subi[\cdot + K]}, \hat{\mathbb{W}}_{i,\subplusk}) \\
            & \leq \VNi(x_{i, t+K-1},  \opt{r}_{\subi[\cdot + K - 1]}, \mathbb{W}_{i,\cdot + K - 1 | t}) \\
            &\quad + \underbrace{ \left( \frac{\alpha_{2,i}}{\alpha_{1,i}} \mathcal{L}_{f_i}^2 \gammabari \left( \frac{\gammabari - 1}{\gammabari} \right)^{N_i} - 1 \right) }_{\eqqcolon \rho_{N,i}} \\
            & \quad \times \ell_i(\optt{x_{i, 0 | t+K-1}}, \optt{u_{i, 0 | t+K-1}},  \opt{r}_{\subi[K - 1]}) \\
            & \leq (1 + \rho_{N,i}) \VNi(x_{i, t+K-1},  \opt{r}_{\subi[\cdot + K - 1]}, \mathbb{W}_{i,\cdot + K - 1 | t}).
        \end{split}
    \end{align*}
    By inspection, it holds that $1 + \rho_{N,i} \leq  1 - \frac{\alpha_{N,i}}{\gammabari}$ for
    \begin{align}
    \label{eq:n_star_lower_bound}
    N_i \geq \opt{N}_i \coloneqq \frac{\ln \left( \frac{\alpha_{2,i}}{\alpha_{1,i}} \mathcal{L}_{f_i}^2 \gammabari^2 \right) - \ln (\gammabari - \alpha_{N,i})}{\ln \gammabari - \ln (\gammabari - 1)}.
    \end{align}
    Then, using \eqref{eq:fast_coverage_exp_cost_decrease_K_steps} for $K-1$ steps, we get
    \begin{align*}
        &\VNi(x_{i, t+K},  \opt{r}_{\subi[\cdot + K]}, \hat{\mathbb{W}}_{i,\subplusk})\\
        &\quad \stackrel{\eqref{eq:fast_coverage_exp_cost_decrease_K_steps}}{\leq}  (1 + \rho_{N,i}) \left( 1 - \frac{\alpha_{N,i}}{\gammabari} \right)^{K-1}
        \VNi(x_{i, t}, \rti, \mathbb{W}_{i,\sub}).\\
        &\quad \stackrel{\eqref{eq:n_star_lower_bound}}{\leq} 
            \left( 1 - \frac{\alpha_{N,i}}{\gammabari} \right)^K
            \VNi(x_{i, t}, \rti, \mathbb{W}_{i,\sub}),
    \end{align*}

    Finally, considering that  $\VNi(x_{i, t}, \rti, \mathbb{W}_{i,\sub}) \leq V_{\max,i}$ and ${\gammabari} > \alpha_{N,i} > 0$, we get \eqref{eq:V_max_reference_in_new_partitions}.
\end{proof}

\subsection{Proof of Proposition \ref{proposition:finite_time_partition_update}}
\label{appendix:close_enough}
We show that the partition update condition is satisfied in finite time by following similar steps to the proof in \cite[Theorem 2]{rickenbach2023}. 
\begin{proof}
We consider a proof of contradiction, i.e., suppose the condition \eqref{eq:partition_update_condition} is not satisfied for $t' \in \left[t, t + \tau \right]$ with some finite $\tau \geq 0$, and hence condition~\eqref{eq:partition_update_condition} can not be satisfied for a fixed reachable reference. Note that for a periodic reference, i.e., $r_{i,\cdot|t+T} = r_{i,\cdot + T|t}\, \forall t > 0$, the reference is defined for an infinite horizon.

First, we derive a constant $V_\epsilon > 0$, such that $V_{i,t'}\leq V_\epsilon$ implies that conditions \eqref{eq:partition_update_condition} are satisfied. Further, we show that we can find some uniform $\tau \geq 0$ for a fixed, reachable reference and any $V_{i,t}\leq \Vmaxi$ we such that $V_{i,t+\tau} \leq V_\epsilon$ $\forall i \in \M$.

\textit{Part 1:} Let $\hat{\W}_{\cdot|t'+1} := \W_{\opt{p}_{\cdot|t'+1}}$. Condition \eqref{eq:partition_update_condition} requires the candidate sequence $(\hat{x}_{i,\cdot|t'+1}, \hat{u}_{i,\cdot|t'+1})$ from \eqref{eq:fast_coverage_candidate_sequence} to lie within the candidate partition sequence $\hat{\Bar{\W}}_{i,\cdot|t'+1}$. According to Lemma~\ref{lemma:feasible_partition_update}, this is ensured if the candidate sequence is close enough to the reference trajectory, i.e., 
\begin{equation}
\label{eq:close_enough}
    \mynorm{ \hat{p}_{i,k|t'+1} - \opt{p_{i,k|t'+1}}} \leq\|C_i\| \mynorm{ \hat{x}_{i,k|t'+1} - \opt{x_{i,k|t'+1}}} \leq \epsilon
\end{equation}
for $k = 0,\ldots, N_i$. Next, we derive a constant $V_\epsilon > 0$ such that this condition holds if $V_{i,t'} \leq V_\epsilon$. For this, we compute the following upper bound.
\begin{align} \label{eq:coverage_mpc_candidate_bound}
    \begin{split}                 
        %&\JNi(\hat{x}_{i,\cdot|t'+1}, \hat{u}_{,\cdot|t'+1}, \opt{r}_{i,\cdot|t'+1})\\
        & \sum_{k=0}^{N_i} \ell_i(\hat{x}_{i,k|t'+1}, \hat{u}_{i,k|t'+1}, \opt{r}_{i,k|t'+1})\\
        & \stackrel{\eqref{eq:fast_coverage_candidate_sequence}}{=} \; \sum_{k=1}^{N_i-2} \ell_i(\optt{x_{i,k|t'}}, \optt{u_{i,k|t'}}, \opt{r_{i,k|t'}}) \\
        & \quad + \ell_i(\optt{x_{i,N_i-1|t'}}, \opt{u_{i,N_i-1|t'}}, \opt{r_{i,N_i-1|t'}}) \\
        & \quad + \ell_i(\hat{x}_{i,N_i -1 | t' + 1}, \opt{u_{i,N_i | t'}}, \opt{r_{i,N_i | t'}}) \\
        & \quad + \ell_i(\hat{x}_{i,N_i | t' + 1}, \opt{u_{i,N_i+1 | t'}}, \opt{r_{i,N_i+1 | t'}}),  \\
        & \leq \JNi(\optt{x_{i, \cdot | t'}}, \optt{u_{i, \cdot | t'}}, \opt{r_{i, \cdot | t'}})\\
        & \quad + \ell_i(\hat{x}_{i,N_i -1 | t' + 1}, \opt{u_{i,N_i | t'}}, \opt{r_{i,N_i | t'}}) \\
        & \quad + \ell_i(\hat{x}_{i,N_i | t' + 1}, \opt{u_{i,N_i+1 | t'}}, \opt{r_{i,N_i+1 | t'}}),\\
    \end{split}
\end{align}
where we denoted $\hat{u}_{i,N_i|t'+1}:=\opt{u_{i,N_i+1|t'}}$.
In particular, the last two terms in \eqref{eq:coverage_mpc_candidate_bound} can be upper bounded using \eqref{eq:bounds_on_stage_cost} and Lipschitz continuity of $f_i$ from Assumption~\ref{assumption1}:  
\begin{align*}
\begin{split}
    &\ell_i(\hat{x}_{i,N_i -1 | t' + 1}, \opt{u_{i,N_i | t'}}, \opt{r_{i,N_i | t'}}) \\
        & \quad \quad \stackrel{\eqref{eq:bounds_on_stage_cost}, \,\mathrm{Asm.}  \ref{assumption1}}{\leq} \alpha_{2,i} \mathcal{L}_{f_i}^2 \mynorm{\optt{x_{i,N_i-1|t'}} - \opt{x_{i,N_i-1|t'}}}^2, \\
    &\ell_i(\hat{x}_{i,N_i | t' + 1}, \opt{u_{i,N_i+1 | t'}}, \opt{r_{i,N_i+1 | t'}})\\
        & \quad \quad \stackrel{\eqref{eq:bounds_on_stage_cost}, \, \mathrm{Asm.}  \ref{assumption1}}{\leq} \alpha_{2,i} \mathcal{L}_{f_i}^4 \mynorm{\optt{x_{i,N_i-1|t'}} - \opt{x_{i,N_i-1|t'}}}^2.
\end{split}
\end{align*}
Then, using $\VNit = \JNi( \optt{\xti}, \optt{\uti}, \opt{\rti} )$ and
\begin{align*}
\begin{split}
    &\mynorm{\optt{x_{i,N_i-1|t'}} - \opt{x_{i,N_i-1|t'}}}^2\\
    &\quad \stackrel{\eqref{eq:bounds_on_stage_cost}}{\leq} \frac{1}{\alpha_{1,i}} \mynorm{\optt{x_{i,N_i-1|t'}} - \opt{x_{i,N_i-1|t'}}}_{Q_i}^2
    \leq \frac{1}{\alpha_{1,i}} V_{i,t'},
\end{split}
\end{align*}
we obtain
\begin{equation}
\label{eq:sum_li_candidate_bound}
    %\JNi(\hat{x}_{i,\cdot|t'+1}, \hat{u}_{i,\cdot|t'+1}, \opt{r}_{i,\cdot|t'+1})
    \sum_{k=0}^{N_i} \ell_i(\hat{x}_{i,k|t'+1}, \hat{u}_{i,k|t'+1}, \opt{r}_{i,k|t'+1})
    \leq (1 + \beta_i) V_{i,t'},
\end{equation}
where $\beta_i = \frac{\alpha_{2,i}}{\alpha_{1,i}} (\mathcal{L}_{f_i}^2 + \mathcal{L}_{f_i}^4)$. 
Consequently, with \eqref{eq:bounds_on_stage_cost} we have for $k = 0,\ldots,N_i$,
\begin{align}
\begin{split}
    &\mynorm{ \hat{x}_{i,k|t'+1} - \opt{x_{i,k|t'+1}}}^2\\
    & \quad \stackrel{\eqref{eq:bounds_on_stage_cost}}{\leq} \frac{1}{\alpha_{1,i}} \ell_i(\hat{x}_{i,k|t'+1}, \hat{u}_{i,k|t'+1}, \opt{r}_{i,k|t'+1}) \\
    & \quad \stackrel{\eqref{eq:sum_li_candidate_bound}}{\leq} \frac{(1+\beta_i) V_{i,t'}}{\alpha_{1,i}} \leq \frac{\epsilon^2}{\|C_i\|^2}.
\end{split}
\end{align}

    The last inequality ensures \eqref{eq:close_enough} if
    \begin{equation}
        V_{i,t'} \leq V_{\epsilon} := \frac{\alpha_{1,i} \epsilon^2}{(1 + \beta_i) \|C_i\|^2}.
    \end{equation}
\textit{Part 2:}
    Inequality~\eqref{eq:fast_coverage_exp_cost_decrease_K_steps} from Theorem~\ref{theorem:stability_of_tracking_mpc} shows that the value function of the tracking MPC exponentially converges to a fixed, reachable reference. Thus, for any $V_{\epsilon} > 0$ and any $\VNit \leq V_{\max,i} \, \forall \, t \geq 0$, there is a sufficiently large $\tau$ which implies that $V_{i,t+\tau} \leq V_\epsilon$, i.e.
    \begin{align*}
    V_{i, t + \tau} &\leq \left( 1 - \frac{\alpha_{N,i}}{\gammabari} \right)^{\tau} V_{i, t} \\
     &\leq \left( 1 - \frac{\alpha_{N,i}}{\gammabari} \right)^{\tau} \Vmaxi \\
    &\leq  \max_{i \in \M} \left( 1 - \frac{\alpha_{N,i}}{\gammabari} \right)^{\tau} \Vmaxi \leq V_{\epsilon}, \label{eq:lyapunov_step_with_max}
\end{align*}
    with ${\gammabari} > \alpha_{N,i} > 0$. The last inequality holds $\forall i \in \M$ for
\begin{equation} \label{eq:tau_solution}
    \tau = \frac{\ln V_{\epsilon} - \ln \left( \max_{i \in \M} \Vmaxi \right)}{\ln \left ( \max_{i \in \M} 1 - \frac{\alpha_{N,i}}{\gammabari} \right)}.
\end{equation}

Consequently, condition $V_{i,t'} \leq V_\epsilon$ holds for $t' = t + \tau$, and it follows from \textit{Part I} that condition \eqref{eq:partition_update_condition} is satisfied, resulting in a feasible partition update. This outcome contradicts the initial assumption, implying that condition \eqref{eq:partition_update_condition} is fulfilled for some $t' \in \left[t, t + \tau \right]$ with the finite uniform bound $\tau$. 
\end{proof}

\begin{IEEEbiography}
[{\includegraphics[width=1in,height=1.25in,clip,keepaspectratio]{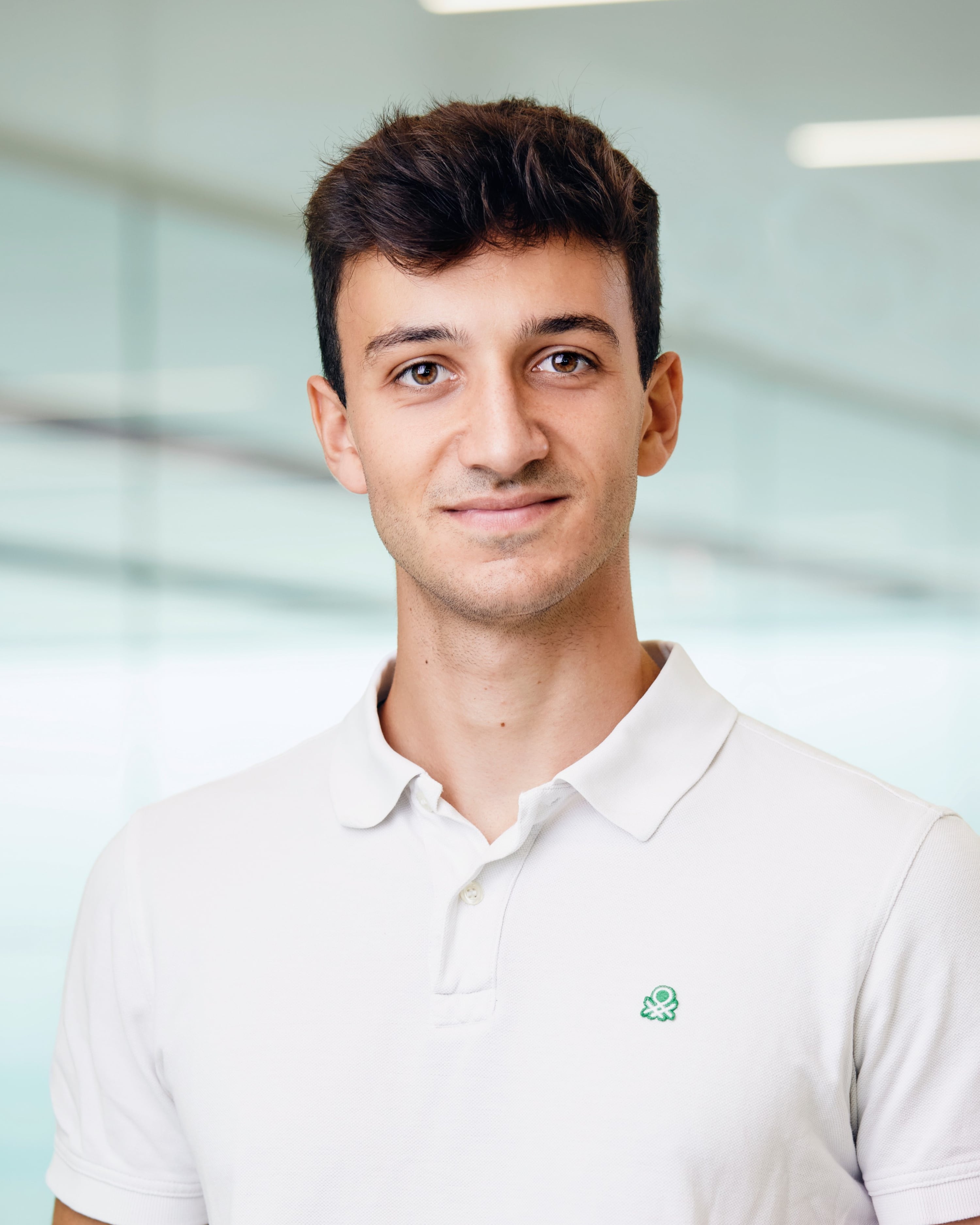}}]{Patrick Benito Eberhard}is currently pursuing the master's degree in Robotics, Systems, and Control at ETH Zurich, Switzerland. He is a member of the Institute for Dynamic Systems and Control and is currently a Visiting Researcher at Stanford University. His research interests include model predictive control and learning-based control for autonomous and multi-agent systems. He is a recipient of the Excellence Scholarship and Opportunity Programme from ETH Zurich.
\end{IEEEbiography}

\begin{IEEEbiography}[{\includegraphics[width=1in,height=1.25in,clip,keepaspectratio]{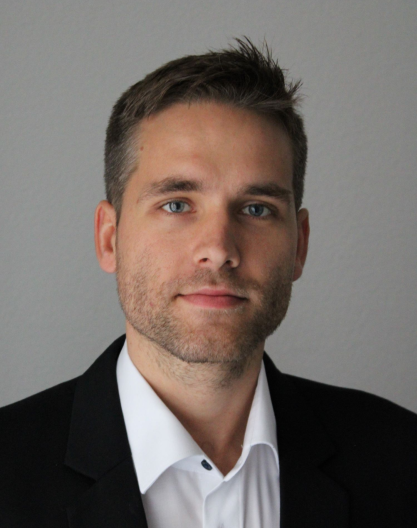}}]{Johannes K\"ohler}
received the Ph.D. degree from the University of Stuttgart, Germany, in 2021. He is currently a postdoctoral researcher at ETH Zürich, Switzerland. 
He is the recipient of the 2021 European Systems \& Control
PhD Thesis Award, the IEEE CSS George S. Axelby
Outstanding Paper Award 2022, and the Journal of
Process Control Paper Award 2023.
His research interests are in the area of model predictive control and control and estimation for nonlinear uncertain systems. 
\end{IEEEbiography}

\begin{IEEEbiography}
[{\includegraphics[width=1in,height=1.25in,clip,keepaspectratio]{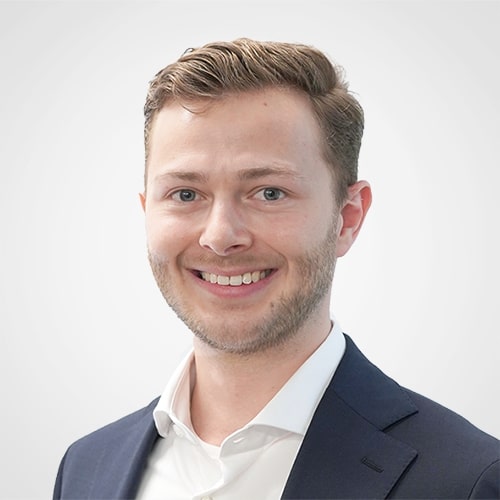}}]{Oliver H\"usser} received his master's degree in electrical engineering in 2024 from ETH Zürich, Switzerland. During his studies, he focused on optimization-based control. His interests included feedback optimization and model predictive control. He is currently working as a consultant at Eraneos, where he strives to lead projects with many unknown parameters and disturbances to success.
\end{IEEEbiography}

\begin{IEEEbiography}[{\includegraphics[width=1in,height=1.25in,clip,keepaspectratio]{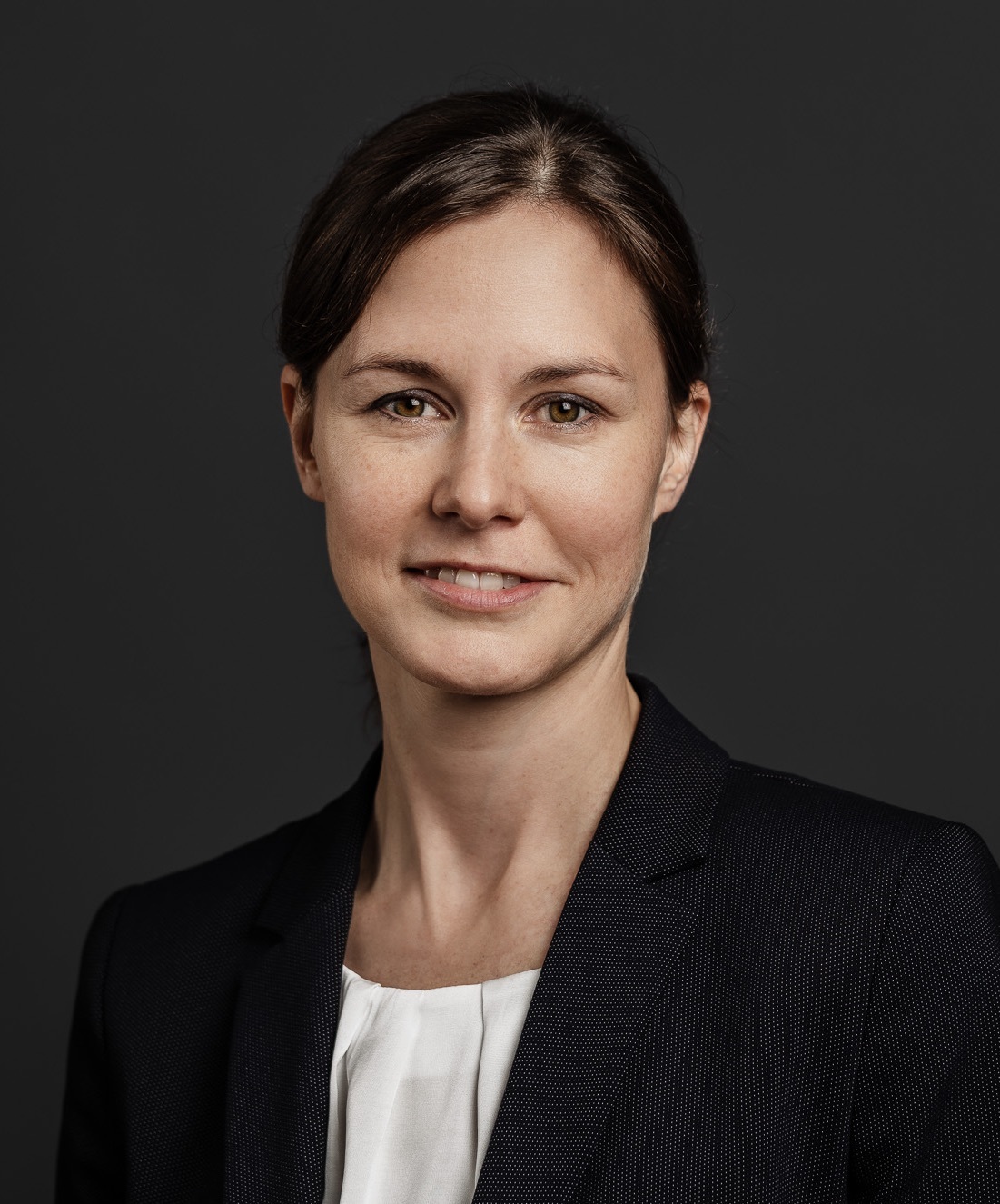}}]{Melanie N. Zeilinger}
 is an Associate Professor at ETH Zürich, Switzerland. 
She received the Diploma degree in engineering cybernetics from the University of Stuttgart, Germany, in 2006, and the Ph.D. degree with honors in electrical engineering from ETH Zürich, Switzerland, in 2011. 
From 2011 to 2012 she was a Postdoctoral Fellow with the Ecole Polytechnique Federale de Lausanne (EPFL), Switzerland.
She was a Marie Curie Fellow and Postdoctoral Researcher with the Max Planck Institute for Intelligent
Systems, Tübingen, Germany until 2015 and with the Department of Electrical Engineering and Computer Sciences at the University
of California at Berkeley, CA, USA, from 2012 to 2014. 
From 2018 to 2019 she was a professor at the University of Freiburg, Germany. 
Her current research interests include safe learning-based control, as well as distributed control and optimization, with applications to robotics and human-in-the loop control.
\end{IEEEbiography}
%\vspace{-35.5\baselineskip} % added to avoid strange spacing in between bios
\begin{IEEEbiography}[{\includegraphics[width=1in,height=1.25in,clip,keepaspectratio]{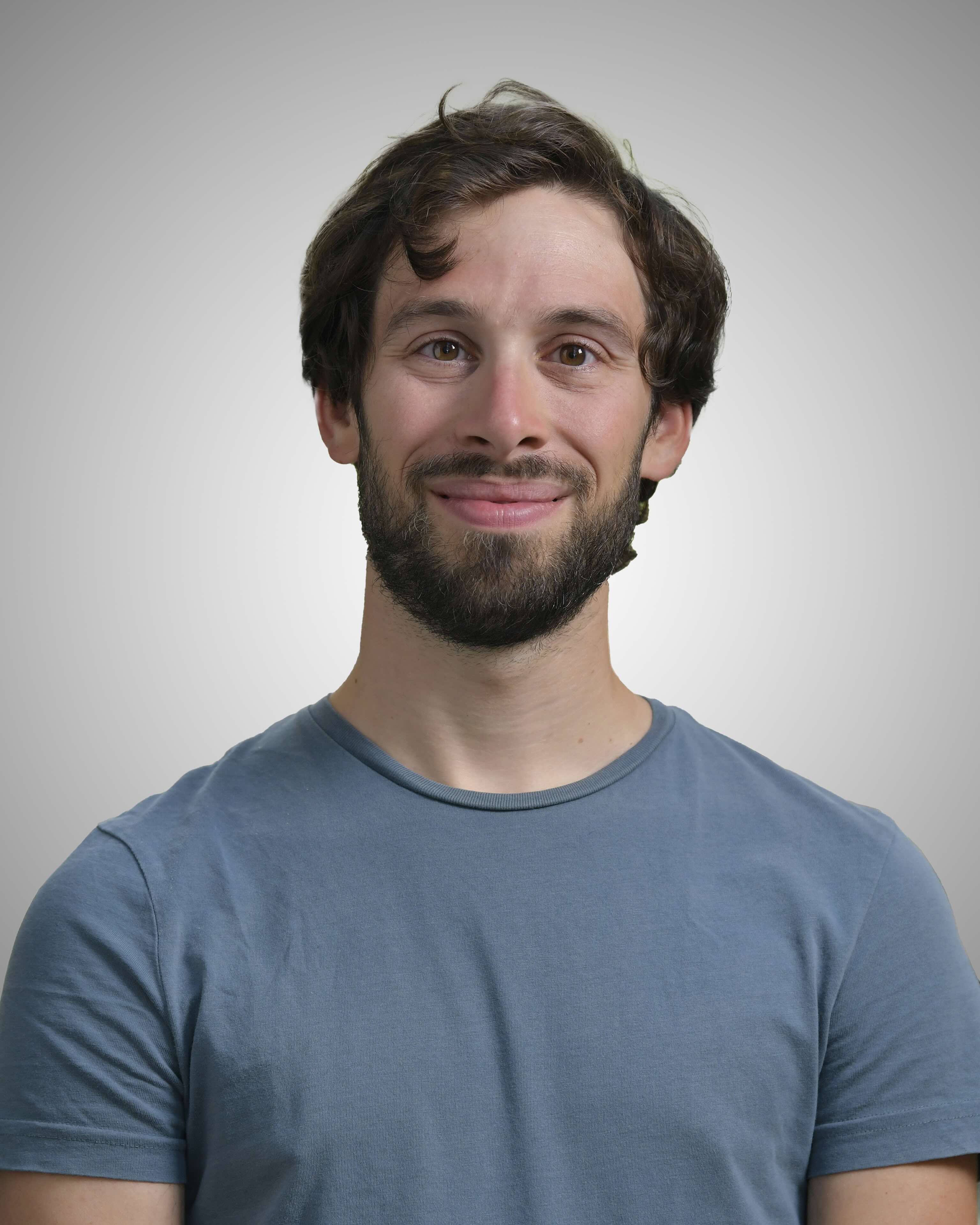}}]{Andrea Carron} received the bachelor’s, master’s, and Ph.D. degrees in control engineering from the University of Padova, Padova, Italy, in 2010, 2012, and, 2016, respectively. He is currently a Senior Lecturer with ETH Zürich. He was a Visiting Researcher with the University of California at Riverside, with Max Planck Institute in Tubingen and with the University of California at Santa Barbara, respectively. From 2016 to 2019, he was a Postdoctoral Fellow with Intelligent Control Systems Group at ETH Zürich. His research interests include safe-learning, learning-based control, multiagent systems, and robotics.
\end{IEEEbiography}
\end{document}